\documentclass[a4paper,11pt]{amsart}

\usepackage[ colorlinks = true,
             linkcolor = blue,
             urlcolor  = blue,
             citecolor = blue,
             anchorcolor = green,
]{hyperref}

\usepackage[top=2cm, bottom=2cm, left=2cm, right=2cm]{geometry}
\usepackage{amsmath,amssymb,amsthm}
\usepackage{graphicx}
\usepackage{mathtools}
\mathtoolsset{showonlyrefs}
\usepackage[sc,osf]{mathpazo}
\usepackage{eulervm}
\usepackage{thmtools, thm-restate}

\newtheorem{theorem}{Theorem}[section]
\newtheorem{lemma}[theorem]{Lemma}
\newtheorem{proposition}[theorem]{Proposition}

\newtheorem{definition}[theorem]{Definition}
\newtheorem{remark}[theorem]{Remark}
\newtheorem{example}[theorem]{Example}
\newtheorem{corollary}[theorem]{Corollary}

\newcommand{\exc}[2]{\underset{#1}{\mathbf{E}}\left\{ #2 \right\}}
\newcommand\blfootnote[1]{%
	\begingroup
	\renewcommand\thefootnote{}\footnote{#1}%
	\addtocounter{footnote}{-1}%
	\endgroup
}

\def\bra #1{\langle #1\vert}
\def\ket #1{\vert #1\rangle}

\def\ketbra #1#2{\vert #1\rangle \langle #2\vert}

\def\tr{{\rm Tr}}

\newcommand*{\eps}{\varepsilon}

\newcommand*{\cI}{\mathcal{I}}

\newcommand*{\cM}{\mathcal{M}}
\newcommand*{\cN}{\mathcal{N}}

\begin{document}

\title{Semidefinite programming hierarchies for constrained bilinear optimization}

\author{Mario Berta}
\address{Department of Computing, Imperial College London, United Kingdom}
\email{m.berta@imperial.ac.uk}

\author{Francesco Borderi}
\address{Department of Computing, Imperial College London, United Kingdom}
\email{f.borderi17@imperial.ac.uk}

\author{Omar Fawzi}
\address{Universit\'e de Lyon, ENS de Lyon, CNRS, UCBL, LIP, F-69342, Lyon Cedex 07, France}
\email{omar.fawzi@ens-lyon.fr}

\author{Volkher B.~Scholz}
\address{Department of Physics, Ghent University, Ghent, Belgium}
\email{volkher.scholz@gmail.com}

\begin{abstract}
We give asymptotically converging semidefinite programming hierarchies of outer bounds on bilinear programs of the form $\mathrm{Tr}\big[H(D\otimes E)\big]$, maximized with respect to semidefinite constraints on $D$ and $E$. Applied to the problem of approximate error correction in quantum information theory, this gives hierarchies of efficiently computable outer bounds on the success probability of approximate quantum error correction codes in any dimension. The first level of our hierarchies corresponds to a previously studied relaxation [Leung \& Matthews, IEEE ITTrans.~2015] and positive partial transpose constraints can be added to give a sufficient criterion for the exact convergence at a given level of the hierarchy.

To quantify the worst case convergence speed of our sum-of-squares hierarchies, we derive novel quantum de Finetti theorems that allow imposing linear constraints on the approximating state. In particular, we give finite de Finetti theorems for quantum channels, quantifying closeness to the convex hull of product channels as well as closeness to local operations and classical forward communication assisted channels. As a special case this constitutes a finite version of Fuchs-Schack-Scudo's asymptotic de Finetti theorem for quantum channels. Finally, our proof methods answer a question of [Brand\~ao \& Harrow, STOC 2013] by improving the approximation factor of de Finetti theorems with no symmetry from $O(d^{k/2})$ to $\mathrm{poly}(d,k)$, where $d$ denotes local dimension and $k$ the number of copies.
\end{abstract}

\maketitle
\blfootnote{*\textit{ Part of this work has been presented at ISIT 2019 under the title "Quantum Coding via Semidefinite Programming".}}
\vspace{-1cm}


\section{Introduction}\label{sec:introduction}

In this paper, we study constrained bilinear optimization problems of the form
\begin{align}\label{eq:form}
Q=\max &\quad \mathrm{Tr}\big[H (D\otimes E)\big]\\
s.t. &\quad D \in\mathcal{P}_D=\Pi_{A \to D}(\mathcal{S}_A^+ \cap \mathcal{A}_{A})\\
&\quad E \in\mathcal{P}_E=\Pi_{B \to E}(\mathcal{S}_B^+ \cap \mathcal{A}_{B}),
\end{align}
where $H$ denotes a matrix in $\mathcal{S}_D\otimes\mathcal{S}_E$ for $\mathcal{S}_D=\mathbb{C}^{d_D\times d_D}$ and $\otimes$ the Kronecker tensor product, and $\mathcal{P}_{D}$ and $\mathcal{P}_{E}$ are positive semidefinite representable sets such that:
\begin{itemize}
\item $\Pi_{A \to D}:\mathcal{S}_A\to\mathcal{S}_D$ and $\Pi_{B \to E}:\mathcal{S}_B\to\mathcal{S}_E$ are linear maps
\item $\mathcal{S}_A^+$ and $\mathcal{S}_B^+$ are the sets of positive semidefinite unit trace matrices in $\mathcal{S}_A$ and $\mathcal{S}_B$, respectively
\item $\mathcal{A}_{A}$ and $\mathcal{A}_B$ are affine subspaces of $\mathcal{S}_A$ and $\mathcal{S}_B$, respectively.
\end{itemize}
Our main motivation to study problems of the form~\eqref{eq:form} comes from quantum information theory\,---\,or more specifically the problem of {\it approximate quantum error correction}. We present this application and its motivation in detail in Section~\ref{sec:quantum-error-correction}, but continue here with the general mathematical setting.

To discuss our approach, we first rewrite~\eqref{eq:form} by defining $G_{AB}= (\Pi_{A \to D}^{\dagger} \otimes \Pi_{B \to E}^{\dagger})(H)$, where $\Pi^{\dagger}$ denotes the adjoint map of $\Pi$ in the Hilbert-Schmidt inner product. This leads to the form
\begin{align}\label{original_program}
Q=\max &\quad \mathrm{Tr}\big[G_{AB} (W_{A} \otimes W_{B})\big]\\
s.t. &\quad W_{A} \succeq 0,\quad W_{B} \succeq 0,\quad\mathrm{Tr}[W_A] = \mathrm{Tr}[W_B] = 1\\
&\quad\Lambda_{A\to C_A}\left(W_{A}\right)=X_{C_A},\quad\Gamma_{B\to C_B}\left(W_{B}\right)=Y_{C_B}, 
\end{align}
where $G_{AB}\in\mathcal{S}_A\otimes\mathcal{S}_B$, $\Lambda_{A\to C_A}:\mathcal{S}_A \to \mathcal{S}_{C_A}$ and $\Gamma_{B\to C_B}:\mathcal{S}_B \to \mathcal{S}_{C_B}$ denote linear maps, and $X_{C_A}\in\mathcal{S}_{C_A}$ and $Y_{C_B}\in\mathcal{S}_{C_B}$ are the matrices defining $\mathcal{A}_{A}$ and $\mathcal{A}_B$ as the affine subspaces associated with the kernels of the linear maps $\Lambda_{A\to C_A}$ and $\Gamma_{B\to C_B}$, respectively.
Now, by the linearity of the objective function we can equivalently optimise over the convex hull of feasible points
\begin{align}\label{eq:transformed_program}
Q=\max &\quad \mathrm{Tr}\left[G_{AB}\left(\sum_{i\in I}p_iW_A^i \otimes W_B^i\right)\right]\\
s.t. &\quad p_i\geq0,\quad W_A^i \succeq 0,\quad W_B^i \succeq 0,\quad\mathrm{Tr}\left[W_A^i\right] = \mathrm{Tr}\left[W_B^i\right] = 1\quad\forall i\in I,\quad\sum_{i\in I}p_i=1\\
&\quad\Lambda_{A\to C_A}\left(W_A^i\right)=X_{C_A},\quad\Gamma_{B\to C_B}\left(W_B^i\right)=Y_{C_B}\quad\forall i\in I.
\end{align}
That is, in the language of quantum information theory we are maximizing over a subset of the so-called {\it separable quantum states}\,---\,where the latter is defined on $A\otimes B$ as
\begin{align}
\text{Sep}(A:B)=\left\{\sum_{i\in I}p_iW_A^i \otimes W_B^i: W_A^i \succeq 0, W_B^i \succeq 0, \mathrm{Tr}\left[W_A^i\right]=\mathrm{Tr}\left[W_B^i\right] = 1,p_i\geq0,\sum_{i\in I}p_i=1\right\}.
\end{align}
Recall that matrices $W_A\in\mathcal{S}^+_A$ are called quantum states on system $A$\,---\,and similarly for bipartite states on $A\otimes B$.

Now, to approximate the set of separable states within the set of bipartite states is a ubiquitous but hard problem in quantum information theory (see, e.g., \cite{Barak12}). Nevertheless, as realized in~\cite{Doherty04} the set of separable states can be approximated by the sum-of-squares hierarchies of Lasserre~\cite{Lasserre00} and Parrilo~\cite{Parrilo2003}. This lead to the {\it semidefinite programming hierarchy} of Doherty-Parrilo-Spedalieri (DPS), which is extensively employed to characterize quantum correlations in quantum information theory~\cite{Doherty02}. The underlying idea of the DPS hierarchy is that separable states $W_{AB}=\sum_i p_i W_A^i \otimes W_B^i$ on $A\otimes B$, where $\{p_i\}_{i\in I}$ is a probability distribution, are {\it $n$-extendible} to $W_{AB_1^n}=\sum_i p_iW_A^i \otimes (W_B^i)^{\otimes n}$ on $A\otimes B^{\otimes n}$ for any $n$, such that we have for any permutation $\pi$ that\footnote{Here and henceforth we use the notation $B_i^j$ to denote the systems $B_i\otimes\cdots\otimes B_j$, which should be interpreted as empty if $i>j$.}
\begin{align}
W_{AB_1^n}=\left(\mathcal{I}_A\otimes\mathcal{U}_{B_1^n}^\pi\right)(W_{AB_1^n})
\end{align}
with $\mathcal{I}_A$ the identity map on $\mathcal{S}_A$ and $\mathcal{U}_{B_1^n}^\pi$ the unitary map that permutes the systems $B_1^n$ according to $\pi\in\mathfrak{S}_n$\,---\,the symmetric group of $n$ elements. The state $W_{AB_1^n}=\sum_i p_iW_A^i \otimes (W_B^i)^{\otimes n}$ is an \textit{extension} of $W_{AB}$, meaning that we again get back the original state $W_{AB}$ when throwing away all additional systems $B_2^n$: $\mathrm{Tr}_{B_2^n}(W_{AB_1^n})=W_{AB}$.\footnote{We refer to Section \ref{sec:notation} for the formal definition of the so-called partial trace map $\mathrm{Tr}_{B_2^n}[\cdot]$.} Due to the monogamy of quantum correlations, however, general states do not have this property ~\cite{Coffman00},~\cite{Hreview2009}. In fact, finite {\it quantum de Finetti theorems} quantify, with upper bounds, the distance of $n$-extendible states to separable states~\cite{Christandl2007}, with convergence in the limit $n\to\infty$~\cite{STORMER196948}. More precisely, \cite[Theorem II.7]{Christandl2007} gives that for states $W_{AB}$ $n$-extendible to $W_{AB_1^n}$, there exists a probability distribution $\{p_i\}_{i\in I}$ and states $W^i_A,W_B^i$ on $A$ and $B$, respectively, such that
\begin{align}\label{eq:intro-deFinetti}
\left\|W_{AB}-\sum_{i\in I}p_iW_A^i\otimes W_B^i\right\|_1\leq\frac{2d_B^2}{n},
\end{align}
where $\|X\|_1=\mathrm{Tr}\left[|X|\right]$ denotes the Schatten one-norm and $d_B$ the dimension of $B$. Crucially, $n$-extendability has a {\it semidefinite representation} and this then immediately gives efficient semidefinite approximations of the set $\text{Sep}(A:B)$ for any fixed $n$.

For our setting, however, we are interested more generally in characterizing bipartite states that are separable, but subject to linear constraints on the $W_A^i,W_B^i$ as well. As such, the approach we use to generate convergent semidefinite programming hierarchies for the constrained bilinear optimizations~\eqref{eq:transformed_program} is based on deriving finite de Finetti representation theorems with additional linear constraints. This leads to our main finding, the semidefinite programs
\begin{align}\label{eq:intro-main}
\mathrm{SDP}_n=\max &\quad \mathrm{Tr}\big[G_{AB} W_{AB}\big]\\
s.t. 
&\quad W_{AB_1^n}\succeq0, \mathrm{Tr}(W_{AB_1^n}) = 1, \;W_{AB_1^n}=\left(\mathcal{I}_A\otimes\mathcal{U}_{B_1^n}^\pi\right)\left(W_{AB_1^n}\right)\;\forall\pi\in\mathfrak{S}_n\\
& \quad \left(\Lambda_{A\to C_A}\otimes\mathcal{I}_{B_1^n}\right)\left(W_{AB_1^n}\right)=X_{C_A}\otimes W_{B_1^n},\quad \left(\mathcal{I}_{B_1^{n-1}}\otimes\Gamma_{B_n\to C_B}\right)\left(W_{B_1^n}\right)=W_{B_1^{n-1}}\otimes Y_{C_B}
\end{align}
form a sequence of upper bounds on $Q$ with the property
\begin{align}
0 \leq \mathrm{SDP}_n - Q \leq \frac{\mathrm{poly}(d)}{\sqrt{n}}\quad\text{implying}\quad Q=\lim_{n\to\infty}\mathrm{SDP}_n,
\end{align}
where $d=\max\{d_A,d_B\}$ and $\mathrm{poly}(d)$ denotes a term at most polynomial in $d$. Notice that the state $W_{AB}$ appearing in the objective function of \eqref{eq:intro-main}, is the reduced state of $W_{AB_1^n}$ on $A\otimes B_1$, i.e., $\mathrm{Tr}_{B_2^n}(W_{AB_1^n})=W_{AB}$.

The remainder of our manuscript is structured as follows. In Section~\ref{sec:deFinetti} we give quantum de Finetti theorems with linear constraints and in Section~\ref{sec:outer-bounds} we present how these lead to an outer hierarchy of converging SDP relaxations for constrained bilinear optimization of the form \eqref{eq:form}. In Section~\ref{sec:quantum-error-correction}, we then discuss as a special case de Finetti theorems for quantum channels (Section~\ref{sec:separable-deFinetti}), which we utilise for our main application about approximate quantum error correction (Section~\ref{subsectionHierarchy}). We end with some conclusions in Section~\ref{sec:conclusion}. Some arguments and extended material are deferred to appendices, which includes some basic numerical studies in Appendix~\ref{sec:numerics}.

We should mention that in recent work, optimization problems similar to \eqref{eq:form} and termed {\it jointly constrained semidefinite bilinear programming} were studied in~\cite{Huber18}, where it was pointed out that they appear in various forms throughout quantum information theory. We notice that the approach in~\cite{Huber18} is based on non-commutative extensions of the classical branch-and-bound algorithm from~\cite{Khayyal83} and is complementary to ours. Another remark is that we should distinguish the setting~\eqref{original_program} studied here from our previous work on quantum bilinear optimization~\cite{Berta16}, where we were interested in bilinear optimizations of the form
\begin{align}\label{eq:previous}
\max &\quad \sum_{\alpha,\beta}G_{\alpha,\beta}\langle\psi|E_\alpha D_\beta|\psi\rangle\\
s.t. &\quad \text{$|\psi\rangle\in\mathcal{H}$: Hilbert space}\\
& \quad \text{$E_\alpha,D_\beta$ Hermitian with $[E_\alpha,D_\beta]=0$},
\end{align}
where $E_{\alpha}$ and $D_\beta$ are operators acting on $\mathcal{H}$ subject to polynomial constraints given by the set of conditions $\{[E_{\alpha}, D_{\beta}]=0\}_{\alpha,\beta}$ expressed by commutators, i.e., $[E_{\alpha}, D_{\beta}]=E_{\alpha}D_{\beta}-D_{\beta}E_{\alpha}$.  Note that in this latter setting~\eqref{eq:previous} the dimension of the underlying Hilbert space is unbounded and optimized over as well \cite{Navascues}. In contrast, for our optimisation~\eqref{original_program} the dimension of the variables is fixed in advance. As such, the scope of applications of our current work is different.


\section{De Finetti theorems with linear constraints}\label{sec:deFinetti}

\subsection{Notation}\label{sec:notation}

In the following, we introduce some notation that is standard in quantum information theory. A $d_A$-dimensional quantum system (or in short {\it system}) is given by an inner product space $\mathbb{C}^{d_A}$ and denoted by $A$. {\it Quantum states} (or in short {\it states}) on $A$ are matrices\footnote{Here and henceforth we use the symbol $:=$ as {\it equal by definition}.}
\begin{align}
\text{$W_A\in\mathcal{S}_A:=\mathbb{C}^{d_A\times d_A}$ with $\mathrm{Tr}[W_A]=1$ and $W_A\succeq0$,}
\end{align}
where $\succeq$ denotes the operator (Loewner) order. Quantum states of rank one are called {\it pure} and can be written as $W_A=\ket{\psi}\bra{\psi}_A$, where $\ket{\psi}_A\in\mathbb{C}^{d_A}$ and $\ket{\psi}\bra{\psi}_A\in\mathcal{S}_A$ denotes the rank-one projector on the vector $\ket{\psi}_A$.

A {\it bipartite system} $AB:=A\otimes B$ is given by an inner product space $\mathbb{C}^{d_A}\otimes\mathbb{C}^{d_B}$, where $\otimes$ denotes the Kronecker tensor product. Correspondingly, states on $AB$ are matrices $W_{AB}\in\mathcal{S}_A\otimes\mathcal{S}_B$ with $\mathrm{Tr}[W_{AB}]=1$ and $W_{AB}\succeq0$. {\it Separable states} are states on $AB$ that are in the convex hull of product states $W_A\otimes W_B$, with $W_A$ and $W_B$ states on $A$ and $B$, respectively. The {\it maximally entangled state} $\Phi_{AB}:=\ket{\Phi}\bra{\Phi}_{AB}$ on $AB$ for $d:=d_A=d_B$ is not separable and defined via
\begin{align}
\text{$\ket{\Phi}_{AB}:=\frac{1}{\sqrt{d}}\sum_{x=1}^d\ket{x}_A\otimes\ket{x}_B$ for some orthonormal basis $\{\ket{x}\}_{x=1}^d$ of $\mathbb{C}^d$.}
\end{align}
The {\it swap operator} $F_{AB}$ on $A\otimes B$ exchanges the two quantum systems, i.e., $F_{AB}(W_A\otimes W_B)=W_B\otimes W_A$ for every state $W_A$ and $W_B$ on $A$ and $B$, respectively. {\it Classical-quantum states} are bipartite states that can be written in the form
\begin{align}
W_{XB}=\sum_{x=1}^{d_X}p_x\ket{x}\bra{x}_X\otimes W_B^x
\end{align}
for a probability distribution $\{p_x\}_{x=1}^{d_A}$, an orthonormal basis $\{\ket{x}_X\}_{x=1}^{d_X}$ of $\mathbb{C}^{d_X}$, and states $W_B^x$ on $B$ for $x=1,\dots,d_X$. We refer to $X$ as the classical part of the bipartite classical-quantum system $XB$.

{\it Quantum channels} (or in short {\it channels}) are linear maps $\mathcal{N}_{A\to B}:\mathcal{S}_A\to\mathcal{S}_B$ that are trace preserving and completely positive\footnote{A linear map $\mathcal{N}_{A\to B}:\mathcal{S}_A\to\mathcal{S}_B$ is said to be completely positive if $\mathcal{N}_{A\to B}\otimes \mathcal{I}_C$ is a positive map for every quantum system $C$, where $\mathcal{I}_C$ denotes the identity map on $\mathcal{S}_C$} (cp). In particular, they map states from the input system $A$ to states on the output system $B$. We often abbreviate bipartite channels $\mathcal{I}_A\otimes\mathcal{N}_B$ that act trivially on the $A$-system as $\mathcal{N}_B$, where $\mathcal{I}_A$ denotes the identity channel on $\mathcal{S}_A$. The partial trace $\mathrm{Tr}_B[\cdot]$ is a channel from $AB$ to $A$ defined via
\begin{align}
\mathrm{Tr}_B[\cdot]:=\sum_{x=1}^{d_B}\big(1_A\otimes\bra{x}_B\big)(\cdot)\big(1_A\otimes\ket{x}_B\big),
\end{align}
where $1_A$ denotes the identity matrix on $A$, and $\{\ket{x}\}_{x=1}^{d_B}$ an orthonormal basis of $\mathbb{C}^{d_B}$. For bipartite states $W_{AB}$, we write $W_A=\mathrm{Tr}_B[W_{AB}]$ for the reduced state on $A$. {\it Quantum measurements} (or in short {\it measurements}) are a special case of channels that can be written in the form
\begin{align}
\mathcal{M}_{A\to I}(\cdot):=\sum_{i\in I}\mathrm{Tr}\left[M_A^i(\cdot)\right]\ket{i}\bra{i}_I
\end{align}
with an orthonormal basis $\{\ket{i}_I\}_{i\in I}$ and $M_A^i\succeq0$ $\forall i\in I$ with $\sum_{i\in I}M_A^i=1_A$.

The {\it Choi-Jamio\l{}kowski isomorphism} relates channels with states. For a channel $\mathcal{N}_{A\to B}$, its Choi state is given by
\begin{align}\label{eq:Choi-state}
J^{\mathcal{N}}_{BA'}:=(\mathcal{N}_{A\to B}\otimes\mathcal{I}_{A'})(\Phi_{AA'}),
\end{align}
where $d_{A'}:=d_A$. Note that $J^{\mathcal{N}}_{A'}=\frac{1_{A'}}{d_{A'}}$. Vice versa, for a bipartite state $W_{A'B}$ with $W_{A'}=\frac{1_{A'}}{d_{A'}}$, its Choi channel is given as
\begin{align}\label{eq:Choi-channel}
\mathcal{N}^W_{A\to B}:Z_A\mapsto d_A\cdot\mathrm{Tr}_{A'}\left[W_{A'B}(Z_{A'}^T\otimes1_B)\right],
\end{align}
where the transpose $T$ is taken with respect to the orthonormal basis of the maximally entangled state in~\eqref{eq:Choi-state}.

The distance between states is quantified by the metric induced from the Schatten one-norm $\|X\|_1:=\mathrm{Tr}\left[|X|\right]$. The distance between channels is quantified by the metric induced from the Diamond norm
\begin{align}
\|\mathcal{N}_{A\to B}\|_{\Diamond}:=\sup_{\|X\|_1\leq1}\|(\mathcal{N}_{A\to B}\otimes\mathcal{I}_A)(X_{AA})\|_1.
\end{align}

A multipartite state $W_{AB_1^n}$ on $AB_1^n\equiv AB_1\cdots B_n$ is called {\it symmetric with respect to $A$} if
\begin{align}
(\mathcal{I}_A\otimes\mathcal{U}_{B_1^n}^\pi)(W_{AB_1^n})=W_{AB_1^n}\quad\forall\pi\in\mathfrak{S}_n,
\end{align}
where $\mathfrak{S}_n$ denotes the symmetric group of $n$ elements and
\begin{align}
\mathcal{U}_{B_1^n}^\pi(W_{B_1}\otimes\cdots\otimes W_{B_n}):=W_{B_{\pi^{-1}(1)}}\otimes\cdots\otimes W_{B_{\pi^{-1}(n)}}.
\end{align}
A bipartite state $W_{AB}$ is called {\it $n$-extendable} if there exists a multipartite extension $W_{AB_1^n}$, i.e., $\mathrm{Tr}_{B_2^n}(W_{AB_1^n})=W_{AB}$, that is symmetric with respect to $A$.


\subsection{Previous work}

General de Finetti representation theorems state that if a multipartite state on $AB_1^n$ is symmetric with respect to $A$, then the reduced state on the first $k$ systems $AB_1^k$ is close to a separable mixture of independent and identical states for $k$ sufficiently smaller than $n$. De Finetti~\cite{deFinetti1937} first proved for the classical case with $A$ trivial, i.e.~$A=\mathbb{C}$, that if $n = \infty$ and $k$ is finite, then the statement holds exactly. Quantitative finite versions for the classical case were later proven and the state-of-the-art bounds can be found in~\cite{Diaconis80}. In the quantum setting, early works considered the $n=\infty$ setting including~\cite{STORMER196948,Hudson1976,Fannes1988,raggio89,Petz1990} in the mathematical physics community and~\cite{fuchs02} in the quantum information theory community. The first finite quantum de Finetti representation theorem was proved in~\cite{koenig05}. The state-of-the-art bounds from \cite{Christandl2007,koenig09c} show that for multipartite states $W_{AB_1^n}$ symmetric with respect to $A$, we have that
\begin{align}
\left\|W_{AB_1^k}-\sum_{i\in I}p_iW_A^i\otimes \left(W_B^i\right)^{\otimes k}\right\|_1\leq\frac{2kd_B^2}{n},
\end{align}
for a probability distribution $\{p_i\}_{i\in I}$ and states $W^i_A,W_B^i$ on $A$ and $B$, respectively. Note that the special case $k=1$ exactly recovers~\eqref{eq:intro-deFinetti}.


\subsection{Proof methods}

In the following, we provide a brief sketch of our proof ideas. For simplicity we restrict to $k=1$, which is the relevant case for~\eqref{eq:transformed_program}. Namely, we start with a multipartite state $W_{AB_1^n}$ symmetric with respect to $A$ and the goal is to identify constraints such that $W_{AB_1}$ is approximated by a mixture of states of the form
\begin{align}
\text{$W^i_A \otimes W^i_B$ with $\Lambda_{A\to C_A}\left(W_A^i\right)=X_{C_A}$ and $\Gamma_{B\to C_B}\left(W_B^i\right)=Y_{C_B}$.}
\end{align}
The standard approach for proving de Finetti theorems~\cite{Christandl2007} proceeds by measuring the systems $B_1^n$ with the uniform measurement on the symmetric subspace given by $\left\{ \ket{\psi}\bra{\psi}_B^{\otimes n} \right\}_{\psi}$. In this case, the candidate mixture of product states is given by
\begin{align}
\int p\left(\psi \right) d \ket{\psi} W_{A | \psi} \otimes \ket{\psi}\bra{\psi}_B\,,
\end{align}
where the integral is computed with respect to the Haar measure, $p(\psi) d\ket{\psi}$ denotes the probability of outcome $\psi$, and $W_{A|\psi}$ the state on $A$ conditioned on obtaining outcome $\psi$ in the measurement. The problem with this candidate is that, in this mixture, there will in general be many terms where
\begin{align}
\text{$\ket{\psi}\bra{\psi}_B$ is such that $\Gamma_{B\to C_B}\left(\ket{\psi}\bra{\psi}_B\right)\neq Y_{C_B}$.}
\end{align}
One could try to modify the measurement so that we only get $\ket{\psi}\bra{\psi}_B$ that satisfy the desired constraint, but this seems difficult. Instead, we use an alternative approach, where the candidate mixture of product states is chosen differently~\cite{koenig05,brandao13b}. Namely, starting from $W_{AB_1^n}$ a well-chosen measurement on the systems $B_2^n$ with measurement outcomes $z_2^n$ leads to the candidate mixture of product states
\begin{align}\label{eq:this-lemma}
\exc{z_2^n}{W_{A|z_2^n} \otimes W_{B|z_2^n}}\,.
\end{align}
The advantage of this candidate state is that by enforcing the right constraints on the global state $W_{AB_1^n}$, namely the ones in~\eqref{eq:intro-main}, we can ensure that $\Lambda_{A\to C_A}(W_{A|z_2^n})=X_{C_A}$ and $\Gamma_{B\to C_B}(W_{B|z_2^n})=Y_{C_B}$. Note that in order for this strategy to work properly, we need the chosen measurement to be informationally complete\,---\, that is, allowing to estimate the expectation value of arbitrary states\,----\,and have a small distortion in the sense that the loss in distinguishibility resulting from applying the measurement is small.


\subsection{Information-theoretic tools}\label{sec:proof-tools}

The starting point for our proof technique is the use of the chain rule of the conditional mutual information, first used in this context in~\cite{brandao11} and further exploited in~\cite{brandao13b}. More precisely, we will use the quantum relative entropy defined as
\begin{align}
D(\rho\|\sigma):=
\begin{cases} \tr(\rho \log \rho) - \tr(\rho \log \sigma) & \mbox{if } support(\rho)\subseteq support(\sigma) \\ \infty & \mbox{otherwise} \end{cases},
\end{align}
where $\rho$ and $\sigma$ are quantum states and the logarithm is taken with respect to the basis two. Recall that the support of an operator is defined as the orthogonal complement of its kernel. The following lemma, which can be found in~\cite{brandao13b}, says that if some classical systems $Z_1^n$ are symmetric with respect to $A$, then conditioning on $Z_1^m$ for some value of $m$ breaks the correlations between $A$ and $Z_{m+1}$. Before stating the lemma, we introduce notation that will be used throughout the section. For a state $\rho_{A Z}$ with a classical $Z$-system, we write
\begin{align}
\rho_{A|z}:= \frac{\mathrm{Tr}_{Z}\Big[\rho_{AZ} \left(1_{A} \otimes \ket{z}\bra{z}\right)\Big]}{\mathrm{Tr}\Big[\rho_{AZ} \left(1_{A} \otimes \ket{z}\bra{z}\right)\Big]}.
\end{align}
We simply write $\exc{z_{1}^{m}}{\cdot}$ for the expectation over the choices of $z_1^m$ and the distribution will be clear from the context.

\begin{lemma}\cite{brandao13b}\label{lem:conditioning-breaks-correlations}
Let $\rho_{A Z_1^n}$ be a classical-quantum state with the $Z_1^n$-systems classical and $\mathcal{U}^\pi_{Z_1^n}(\rho_{A Z_1^n})=\rho_{A Z_1^n}$ for all $\pi\in\mathfrak{S}_n$. Then, there exists $0\leq m < n$ such that
\begin{align}
\exc{z_{1}^{m}}{D(\rho_{AZ_{m+1}|z_1^m} \| \rho_{A|z_1^m} \otimes \rho_{Z_{m+1}|z_1^m})} &\leq \frac{\log d_A}{n}
\end{align}
as well as
\begin{align}
\exc{z_{1}^{m}}{\| \rho_{AZ_{m+1}|z_1^m} - \rho_{A|z_1^m} \otimes \rho_{Z_{m+1}|z_1^m} \|_1^2} &\leq \frac{(2 \ln 2) \log d_A}{n},
\end{align}
where $\ln(\cdot)$ denotes the natural logarithm.
\end{lemma}

\begin{proof}
For the quantum mutual information we have $I\left(A:Z_1^n\right)_{\rho}:=D(\rho_{AZ_1^n}\|\rho_A\otimes \rho_{Z_1^n}) \leq \log d_A$ as well as (see, e.g., \cite[Chapter 11]{nielsen00})
\begin{align}
I\left(A:Z_1^n\right)_{\rho}=\sum_{m=0}^{n-1} I(A : Z_{m+1} | Z_1^m)_{\rho}
\end{align}
for the quantum conditional mutual information $I(A : Z_{m+1} | Z_1^m)_\rho:=I(A:Z_1^{m+1})_\rho-I(A:Z_1^m)_\rho$. As a result, there exists an $m \in \{0, \cdots, n-1\}$ such that $I(A : Z_{m+1} | Z_1^m)_{\rho} \leq \frac{\log d_A}{n}$, which implies
\begin{align}
\exc{z_1^m}{I(A:Z_{m+1})_{\rho_{AZ_{m+1}|z_1^m}}} \leq \frac{\log d_A}{n},
\end{align}
where we used $I(A : Z_{m+1} | Z_1^m)_{\rho}=\exc{z_1^m}{I(A:Z_{m+1})_{\rho_{AZ_{m+1}|z_1^m}}}$, which holds since the conditioning systems are classical. The second statement then follows directly from Pinsker's inequality $D(\rho\|\sigma)\geq\frac{1}{2\ln 2}\left\|\rho-\sigma\right\|_1^2$ \cite[Theorem 5.38]{watrous-book}.
\end{proof}

To prove the de Finetti theorem, we will crucially make use of so-called informationally complete measurements for which the loss in distinguishability, or {\it distortion}, can be bounded.

\begin{lemma}{\cite[Lemma 14]{brandao13b}}\label{lem:product-ic-meas}
There exist a product measurement $\mathcal{M}_A\otimes\mathcal{M}_B$ with finitely many outcomes such that for any Hermitian and traceless matrix $\xi_{AB}$ on $A\otimes B$, we have 
\begin{align}
\|(\mathcal{M}_{A} \otimes \mathcal{M}_{B})(\xi_{AB})\|_1 \geq \frac{1}{18\sqrt{d_A d_B}}\| \xi_{AB}\|_{1}.
\end{align}
\end{lemma}

This \cite[Lemma 14]{brandao13b} follows from the methods in \cite{lancien13}. More generally, we define the {\it minimal distortion} for the bipartite system $A \otimes B$ as
\begin{align}\label{eq:def-f1}
f(A, B):=\inf_{\cM_{A}, \cM_{B}}\max_{\substack{\xi_{AB}^{\dagger} = \xi_{AB}\\ \xi_A = 0, \xi_B = 0}} \frac{\| \xi_{AB} \|_1}{\| (\cM_{A} \otimes \cM_B)(\xi_{AB})\|_1},
\end{align}
where the infimum is over all product measurements on $AB$. In this notation, Lemma~\ref{lem:product-ic-meas} shows that
\begin{align}
f(A, B) \leq 18 \sqrt{d_A d_B}.
\end{align}
Note that in the definition of $f(A,B)$ we restricted the maximization to matrices satisfying $\xi_{A} = 0$ and $\xi_{B} = 0$ because this is sufficient for us.

A drawback of Lemma~\ref{lem:product-ic-meas} is that the distortion depends on the dimension $d_A$. More generally, we define the {\it minimal distortion with side information} for a system $B$ as
\begin{align}\label{eq:def-f2}
f(B|\cdot):=\inf_{\cM_B} \sup_{\substack{\xi_{AB}^\dagger=\xi_{AB}\\ \xi_{A} = 0, \xi_B = 0}} \frac{\| \xi_{AB} \|_1}{\| (\cI_{A} \otimes \cM_B)(\xi_{AB})\|_1},
\end{align}
where the infimum is over all measurements on $B$ and the supremum is over all finite-dimensional systems $A$. In Lemma~\ref{lem:ic-meas-side-info} we give an elementary proof that
\begin{align}
f(B|\cdot) \leq d_B^2 (d_B+1)
\end{align}
using state two-designs and properties of weighted non-commutative $L_p$-spaces. In fact, after completion of our work we realised that methods from operator space theory even give the stronger bound
\begin{align}
f(B|\cdot) \leq \sqrt{18d_B^3},
\end{align}
which is discussed in \cite[Equation 66]{brandao13c}. We leave it as an open question to determine the exact dimensional dependence of the minimal distortion with side information.


\subsection{Main technical result}

Combining the tools from the previous subsection we find the following de Finetti theorem with linear constraints.

\begin{theorem}\label{thm:main-deFinetti}
Let $\rho_{AB_1^n}$ be a quantum state, $\Lambda_{A\to C_A},\Gamma_{B\to C_B}$ linear maps, and $X_{C_A},Y_{C_B}$ matrices such that
\begin{align}
\mathcal{U}_{B_1^n}^\pi(\rho_{AB_1^n})&=\rho_{AB_1^n}\;\forall\pi\in\mathfrak{S}_n\qquad\text{symmetric with respect to $A$}\\
\Lambda_{A\to C_A}(\rho_{AB_1^n})&=X_{C_A}\otimes\rho_{B_1^n}\qquad\quad\;\;\text{linear constraint on $A$}\\
\Gamma_{B_n\to C_B}(\rho_{B_1^n})&=\rho_{B_1^{n-1}}\otimes Y_{C_B}\qquad\quad\text{linear constraint on $B$}.
\end{align}
Then, we have that
\begin{align}
\left\|\rho_{AB}-\sum_{i\in I}p_i\sigma^i_{A}\otimes\omega^i_B\right\|_1\leq&\;\min\big\{ f(A,B) , f(B|\cdot) \big\}\sqrt{ \frac{(2 \ln 2) \log \left(d_A\right)}{n}}
\end{align}
with $\{p_i\}_{i\in I}$ a probability distribution, $\rho_{AB}=\tr_{B_2^n}\left[\rho_{AB_1^n}\right]$, and quantum states $\sigma^i_A,\omega_B^i$ such that for $i\in I$:
\begin{align}
\Lambda_{A\to C_A}\left(\sigma^i_A\right)=X_{C_A}\quad\text{and}\quad\Gamma_{B\to C_B}\left(\omega^i_B\right)=Y_{C_B}.
\end{align}
As stated in Section~\ref{sec:proof-tools}, we can, e.g., take $f(A,B)\leq18\sqrt{d_A d_B}$ or $f(B|\cdot)\leq\sqrt{18d_B^3}$.
\end{theorem}

\begin{proof}
Let $\cM_{B}$ be a measurement of the $B$ system and call the outcome system $Z$. Consider the state $\rho_{AZ_1^n}$ obtained by measuring all the $B$ systems with $\cM_{B}$. This distribution is symmetric with respect to $A$ and so we can apply Lemma~\ref{lem:conditioning-breaks-correlations}. We find that there exists an $m \in \{0, \cdots, n-1\}$ such that
\begin{align}
\exc{z_1^m}{\| \rho_{AZ_{m+1}|z_1^m} - \rho_{A|z_1^m} \otimes \rho_{Z_{m+1}|z_1^m} \|_1^2} \leq \frac{(2 \ln 2) \log d_{A}}{n}.
\end{align}
Note that we have for any $z_1^m$, $\rho_{AZ_{m+1}|z_1^m} = (\cI_{A} \otimes \cM_{B})(\rho_{AB_{m+1}|z_1^m})$ and correspondingly $\rho_{Z_{m+1}|z_1^m} = \cM_{B}(\rho_{B_{m+1}|z_1^m})$. Now, we choose the measurement $\cM_{B}$ to be as in Lemma~\ref{lem:ic-meas-side-info} and achieving $f(B|\cdot)$ in \eqref{eq:def-f2}, we get that $\| \xi_{AB} \|_1^2 \leq f(B|\cdot)^2 \| (\cI_{A} \otimes \cM_{B})(\xi_{AB}) \|_1^2$, where $\xi_{AB} = \rho_{AB_{m+1}|z_1^m} - \rho_{A|z_1^m} \otimes \rho_{B_{m+1}|z_1^m}$. As a result, we have
\begin{align}
\exc{z_1^m}{\| \rho_{AB_{m+1}|z_1^m} - \rho_{A|z_1^m} \otimes \rho_{B_{m+1}|z_1^m} \|_1^2} \leq f(B|\cdot)^2\frac{(2 \ln 2) \log d_{A}}{n}.
\end{align}
But note we can also choose measurements $\cM_{A}$ and $\cM_{B}$ achieving $f(A,B)$ in \eqref{eq:def-f1}. In this case,
\begin{align} 
\| \rho_{AB_{m+1}|z_1^m} - \rho_{A|z_1^m} \otimes \rho_{B_{m+1}|z_1^m} \|_1^2 &\leq f(A,B)^2\| (\cM_A \otimes \cM_{B})(\rho_{AB_{m+1}|z_1^m} - \rho_{A|z_1^m} \otimes \rho_{B_{m+1}|z_1^m}) \|_1^2 \\
&\leq f(A,B)^2 \| (\cI_A \otimes \cM_{B})(\rho_{AB_{m+1}|z_1^m} - \rho_{A|z_1^m} \otimes \rho_{B_{m+1}|z_1^m}) \|_1^2 \\
&= f(A,B)^2 \| \rho_{AZ_{m+1}|z_1^m} - \rho_{A|z_1^m} \otimes \rho_{Z_{m+1}|z_1^m} \|_1^2 ,
\end{align}
where we used the fact that the trace norm cannot increase when applying the quantum channel $\cM_{A}$ \cite[Theorem 3.39]{watrous-book}. As a result, we get
\begin{align}
\exc{z_1^m}{\| \rho_{AB_{m+1}|z_1^m} - \rho_{A|z_1^m} \otimes \rho_{B_{m+1}|z_1^m} \|_1^2} \leq f(A,B)^2 \frac{(2 \ln 2) \log d_{A}}{n}.
\end{align}
Now, using the convexity of the square function, we get
\begin{align}
\exc{z_{1}^m}{\| \rho_{AB_{m+1}|z_1^m} - \rho_{A|z_1^m} \otimes \rho_{B_{m+1}|z_1^m} \|_1}&\leq \sqrt{\exc{z_{1}^m}{\| \rho_{AB_{m+1}|z_1^m} - \rho_{A|z_1^m} \otimes \rho_{B_{m+1}|z_1^m} \|^2_1}} \\
&\leq  \min\big\{f(A, B), f(B|\cdot)\big\}\sqrt{\frac{(2 \ln 2) \log d_{A}}{n}}.
\end{align}
Then, using the convexity of the norm and the fact that $\exc{z_1^m}{\rho_{AB_{m+1}|z_1^m}} = \rho_{AB_{m+1}}$, we obtain
\begin{align}
\left\| \rho_{AB_{m+1}} - \exc{z_1^m}{\rho_{A|z_1^m} \otimes \rho_{B_{m+1}|z_1^m}} \right\|_1 &\leq \min\big\{f(A,B), f(B|\cdot)\big\}\sqrt{\frac{(2 \ln 2) \log d_{A}}{n}}.
\end{align}
The state $\exc{z_1^m}{\rho_{A|z_1^m} \otimes \rho_{B_{m+1}|z_1^m}}$ corresponds to our candidate mixture of product states. It now remains to show that all the states in the mixture satisfy the linear constraints. Indeed we have for any $z_1^m$, writing $M_{z}$ for matrices of the measurement $\cM_{B}$,
\begin{align}
\Lambda_{A \to C_{A}}(\rho_{A|z_1^m})
&= \frac{\tr_{B_1^m}\Big[(1_{A}\otimes M_{z_1} \otimes \cdots \otimes M_{z_m}) \Lambda_{A \to C_{A}}(\rho_{A B_1^m})\Big]}{\tr\Big[(1_{A}\otimes M_{z_1} \otimes \cdots \otimes M_{z_m})\rho_{A B_1^m}\Big]} \\
&= \frac{\tr_{B_1^m}\Big[(1_{A}\otimes M_{z_1} \otimes \cdots \otimes M_{z_m}) (X_{C_A} \otimes \rho_{B_1^m}) \Big]}{\tr\Big[(1_{A}\otimes M_{z_1} \otimes \cdots \otimes M_{z_m})\rho_{A B_1^m}\Big]} \\
&= X_{C_A},
\end{align}
and similarly
\begin{align}
\Gamma_{B \to C_{B}}(\rho_{B_{m+1}|z_1^m})
&= \frac{\tr_{B_1^m}\Big[(M_{z_1} \otimes \cdots \otimes M_{z_m} \otimes 1_{C_{B}} ) \Gamma_{B_{m+1} \to C_{B}}(\rho_{B_1^{m+1}})\Big]}{\tr\Big[(M_{z_1} \otimes \cdots \otimes M_{z_m} \otimes 1_{B_{m+1}} )\rho_{B_1^{m+1}}\Big]} \\
&= \frac{\tr_{B_1^m}\Big[(M_{z_1} \otimes \cdots \otimes M_{z_m} \otimes 1_{B_{m+1}} ) (\rho_{B_1 \cdots B_{m}} \otimes Y_{C_B})\Big]}{\tr\Big[(M_{z_1} \otimes \cdots \otimes M_{z_m} \otimes 1_{B_{m+1}} )\rho_{B_1^{m+1}}\Big]} \\
&= Y_{C_B}.
\end{align}
\end{proof}

This can then be extended to a full quantum de Finetti theorem for any reduced state $\rho_{AB_1^{k}}$ with $0<k<n$.

\begin{theorem}\label{cor:main-deFinetti}
For the same setting as in Theorem~\ref{thm:main-deFinetti}, we have for $0<k<n$ that
\begin{align}
\left\|\rho_{AB_1^{k} }-\sum_{i\in I}p_i\sigma^i_{A}\otimes\left(\omega^i_B\right)^{\otimes k}\right\|_1\leq kf(B| \cdot )\sqrt{ (2 \ln 2) \frac{ \log d_{A} + (k - 1) \log d_B }{n-k+1}}.
\end{align}
\end{theorem}

\begin{proof}
Note that the for the state $\rho_{AB_1^{k-1} B_{k}^n}$, the systems $B_{k}^n$ are symmetric with respect to $AB_1^{k-1}$. As such, we can apply the same argument used in the proof of Theorem~\ref{thm:main-deFinetti}, but this time starting from the decomposition $I\left(AB_1^{k-1}:Z_k^n\right)_{\rho}=\sum_{m=k-1}^{n-1} I(AB_1^{k-1} : Z_{m+1} | Z_k^m)_{\rho}$, leading to
\begin{align}
\frac{1}{n-k+1} \sum_{m=k}^n \exc{z_{k+1}^m}{\| \rho_{A B_1^k | z_{k+1}^m} - \rho_{A B_1^{k-1}|z_{k+1}^{m}} \otimes \rho_{B_k | z_{k+1}^m} \|_1}\leq f(B|\cdot)\sqrt{\frac{(2 \ln 2) \log (d_{A} d_B^{k-1})}{n-k+1}} .
\end{align}
Similarly, for any $i \in \{1,\dots,k\}$, we have
\begin{align}
\frac{1}{n-k+1} \sum_{m=k}^n \exc{z_{k+1}^m}{\| \rho_{A B_1^i | z_{k+1}^m} - \rho_{A B_1^{i-1}|z_{k+1}^{m}} \otimes \rho_{B_{i} | z_{k+1}^m} \|_1}\leq f(B| \cdot )\sqrt{\frac{(2 \ln 2) \log (d_{A} d_B^{i-1})}{n-k+1}} .\label{eq:deFinetti-general-i}
\end{align}
Now, using the triangle inequality $k-1$ times, we get for any $m \in \{k, \dots, n\}$ and any $z_{k+1}^m$ that 
\begin{align}
&\left\| \rho_{A B_1^k| z_{k+1}^m} - \rho_{A |z_{k+1}^{m}} \otimes \rho_{B_1|z_{k+1}^m} \otimes \cdots \otimes \rho_{B_k |z_{k+1}^{m}}\right\|_1 \\
&\leq \sum_{i=1}^k \Big\| \rho_{A B_1^i| z_{k+1}^m} \otimes \rho_{B_{i+1}|z_{k+1}^{m}} \otimes \cdots \otimes \rho_{B_k|z_{k+1}^{m}}- \rho_{A B_1^{i-1}|z_{k+1}^{m}} \otimes \rho_{B_{i} | z_{k+1}^m} \otimes \rho_{B_{i+1}|z_{k+1}^{m}} \otimes \cdots \otimes \rho_{B_k|z_{k+1}^{m}}\Big\|_1 \\
&= \sum_{i=1}^k\left\| \rho_{A B_1^i | z_{k+1}^m} - \rho_{A B_1^{i-1}|z_{k+1}^{m}} \otimes \rho_{B_{i} | z_{k+1}^m}\right\|_1.
\end{align}
Taking the average over $m$ and $z_{k+1}^m$ and using~\eqref{eq:deFinetti-general-i}, we get
\begin{align}
&\frac{1}{n-k+1} \sum_{m=k}^n \exc{z_{k+1}^m}{\left\| \rho_{A B_1^k| z_{k+1}^m} - \rho_{A |z_{k+1}^{m}} \otimes \rho_{B_1|z_{k+1}^m} \otimes \cdots \otimes \rho_{B_k |z_{k+1}^{m}}\right\|_1} \\
&\leq k f(B|\cdot)\sqrt{\frac{(2 \ln 2) \log (d_{A} d_B^{k-1})}{n-k+1}} .
\end{align}
As a result, there is an $m$ such that the previous inequality holds. Then, as before, we use the convexity of the norm to put the expectation inside, getting the existence of an $m$ such that
\begin{align}
\left\| \rho_{A B_1^k} - \exc{z_{k+1}^m}{\rho_{A |z_{k+1}^{m}} \otimes \rho_{B_1|z_{k+1}^m} \otimes \cdots \otimes \rho_{B_k |z_{k+1}^{m}}}\right\|_1\leq kf(B|\cdot)\sqrt{ (2 \ln 2) \frac{ \log d_{A} + (k - 1) \log d_B }{n-k+1}}.
\end{align}
To conclude, it suffices to observe that by symmetry $\rho_{B_i|z_{k+1}^m} = \rho_{B_{1}|z_{k+1}^m}$ for all $i \in \{1,\dots,k\}$ and the linear constraints are satisfied by the same calculation as in the proof of Theorem~\ref{thm:main-deFinetti}.
\end{proof}


\subsection{De Finetti theorems without symmetries}\label{sec:no-symmetry}

These results can again be strengthened to a form studied in \cite{brandao13b}, where $\rho_{AB_1^n}$ is not assumed to be symmetric but rather the systems that are kept are chosen at random. More precisely, we improve the so-called de Finetti theorem without symmetries of~\cite[Section 3]{brandao13b} by reducing the dependence from $d_B^{k/2}$ to polynomial in both $d_B$ and $k$, thereby solving one of the problems~\cite[Section 9]{brandao13b} had left open.

\begin{theorem}
Let $\rho_{B_1^n}$ be a quantum state with the systems $B_i$ all having dimension $d_B$. Furthermore, let the entries of $\vec{i}=(i_1, \dots, i_k)$, $\vec{j}=(j_1, \dots, j_{n-k})$ be a random permutation of $\{1, \dots, n\}$, and assume we measure the systems $j_1,\dots,j_{n-k}$ each using the measurement $\cM_B$, getting the classical systems $Z_{j_1}, \dots, Z_{j_{n-k}}$. Then, there exists $m \in \{0, \dots, n - k\}$ such that
\begin{align}
\exc{\vec{i}, \vec{j}, z_{j_1}, \dots, z_{j_m} }{\left\|\rho_{B_{\vec{i}}|z_{j_1} \cdots z_{j_m}} - \rho_{B_{i_1}|z_{j_1} \cdots z_{j_m}} \otimes \cdots \otimes \rho_{B_{i_k|z_{j_1} \cdots z_{j_m}}} \right\|_1} &\leq kf(B|\cdot) \sqrt{ (2 \ln 2) \frac{(k-1)\log d_B }{n-k+1}} \\
&\leq\frac{3k^{3/2}d_B^3 \log d_B}{\sqrt{n-k+1}},
\end{align}
where $f(B|\cdot)$ is defined in~\eqref{eq:def-f2}.
\end{theorem}

To compare with the usual de Finetti theorems with symmetry, the expectation is taken inside the trace norm (by convexity)\,---\,which can then be understood as enforcing the permutation invariance of the state.

\begin{proof}
For fixed $\vec{i}, \vec{j}$, $m \in \{0, \dots, n-k\}$, and $z_{j_1} \cdots z_{j_m}$, we have using the triangle inequality $k-1$ times, 
\begin{align}
&\| \rho_{B_{i_1} \cdots B_{i_k} | z_{j_1} \cdots z_{j_m} } - \rho_{B_{i_1}|z_{j_1} \cdots z_{j_m} } \otimes \cdots \otimes \rho_{B_{i_k} | z_{j_1} \cdots z_{j_m} }  \|_1 \\
&\leq \sum_{t=1}^k \Big\| \rho_{B_{i_1 \cdots i_t} | z_{j_1} \cdots z_{j_m} } \otimes \rho_{B_{i_{t+1}}| z_{j_1} \cdots z_{j_m} } \otimes \cdots \otimes \rho_{B_{i_k}| z_{j_1} \cdots z_{j_m} }\\
&\qquad\quad-  \rho_{B_{i_1 \cdots i_{t-1}} | z_{j_1} \cdots z_{j_m} } \otimes \rho_{B_{i_{t}}| z_{j_1} \cdots z_{j_m}} \otimes \cdots \otimes \rho_{B_{i_k}| z_{j_1} \cdots z_{j_m}} \Big\|_1 \\
&= \sum_{t=1}^k \left\| \rho_{B_{i_1 \cdots i_t} | z_{j_1} \cdots z_{j_m} } -  \rho_{B_{i_1 \cdots i_{t-1}} | z_{j_1} \cdots z_{j_m}  } \otimes \rho_{B_{i_{t}} | z_{j_1} \cdots z_{j_m} }\right\|_1. \label{eq:deFinetti-without-sym-triangle-ineq}
\end{align}
Now, consider a fixed $t$ and fixed values for $i_1, \dots, i_{t-1}$, and assume we additionally measure the system $B_{i_t}$ using the measurement $\cM_B$, getting the classical system $Z_{i_t}$. Then, for fixed $i_1, \dots, i_{t-1}$, the resulting distributions on $(i_t,j_1)$ and $(j_1,i_t)$ are identical, and the same holds for $(i_t,j_1,j_2)$ and $(j_2,j_1,i_t)$, and so on. Hence, we find by elementary entropic identities that
\begin{align}
&\exc{i_t, \vec{j}}{I(B_{i_1} \cdots B_{i_{t-1}} : Z_{i_t} Z_{j_1} \cdots Z_{j_{n-k}})_{\rho} }\\
&=\sum_{m=1}^{n-k} \exc{i_t, \vec{j}}{I(B_{i_1} \cdots B_{i_{t-1}} : Z_{i_t})_\rho+I(B_{i_1} \cdots B_{i_{t-1}} : Z_{j_m}| Z_{j_1} \cdots Z_{j_{m-1}}Z_{i_t})_{\rho}}\\
&= \sum_{m=0}^{n-k} \exc{i_t, \vec{j}}{I(B_{i_1} \cdots B_{i_{t-1}} : Z_{i_t}| Z_{j_1} \cdots Z_{j_m})_{\rho}} \\
&= \sum_{m=0}^{n-k} \exc{i_t, \vec{j}, z_{j_1}, \dots, z_{j_m} }{I(B_{i_1} \cdots B_{i_{t-1}} : Z_{i_t})_{\rho_{|z_{j_1} \cdots z_{j_m} } } } .
\end{align}
Note on the other hand that we have $I(B_{i_1} \cdots B_{i_{t-1}} : Z_{i_t} Z_{j_1} \cdots Z_{j_{n-k}} ) \leq \log(d_{B}^{t-1})$ and thus we get, using Pinsker's inequality,
\begin{align}
&\frac{1}{n-k+1} \sum_{m=0}^{n-k} \exc{i_t, \vec{j}, z_{j_1}, \dots, z_{j_m} }{\| \rho_{B_{i_1 \cdots i_{t-1}} Z_{i_t} |z_{j_1} \cdots z_{j_m}} -  \rho_{B_{i_1 \cdots i_{t-1}} |z_{j_1} \cdots z_{j_m}} \otimes \rho_{Z_{i_{t}}|z_{j_1} \cdots z_{j_m}} \|_1^2}\\
&\leq (2 \ln 2) \frac{\log d_B^{t-1}}{n-k+1}.
\end{align}
Observe that $\rho_{B_{i_1 \cdots i_{t-1}} Z_{i_t} |z_{j_1} \cdots z_{j_m}} = \cM_{B_{i_t}}(\rho_{B_{i_1 \cdots i_{t-1}} B_{i_t} |z_{j_1} \cdots z_{j_m}})$ and using a measurement $\cM_{B}$ achieving $f(B|\cdot)$ in~\eqref{eq:def-f2} (or using the measurement in Lemma~\ref{lem:ic-meas-side-info}, in which case we should replace $f(B|\cdot)$ by $d_B^2(d_B+1)$ in the following equations), we get that 
\begin{align}
&\frac{1}{n-k+1} \sum_{m=0}^{n-k} \exc{i_t, \vec{j}, z_{j_1}, \dots, z_{j_m}}{\| \rho_{B_{i_1 \cdots i_{t-1}} B_{i_t} | z_{j_1} \cdots z_{j_m} } -  \rho_{B_{i_1 \cdots i_{t-1}} | z_{j_1} \cdots z_{j_m}} \otimes \rho_{B_{i_{t}}|z_{j_1} \cdots z_{j_m} } \|_1^2} \\
&\leq (2 \ln 2) f(B|\cdot)^2 \frac{\log d_B^{t-1}}{n-k+1}.
\end{align}
This implies, using the convexity of the square function, that 
\begin{align}
&\frac{1}{n-k+1} \sum_{m=0}^{n-k} \exc{i_t, \vec{j}, z_{j_1}, \dots, z_{j_m}}{\| \rho_{B_{i_1 \cdots i_{t-1}} B_{i_t} | z_{j_1} \cdots z_{j_m} } -  \rho_{B_{i_1 \cdots i_{t-1}} | z_{j_1} \cdots z_{j_m}} \otimes \rho_{B_{i_{t}}|z_{j_1} \cdots z_{j_m} } \|_1}\\
&\leq f(B|\cdot)\sqrt{(2 \ln 2) \frac{\log d_B^{t-1}}{n-k+1}},
\end{align}
and we get, continuing on~\eqref{eq:deFinetti-without-sym-triangle-ineq}, that
\begin{align}
&\frac{1}{n-k+1} \sum_{m=0}^{n-k} \exc{\vec{i}, \vec{j}, z_{j_1}, \dots, z_{j_m} }{ \| \rho_{B_{i_1} \cdots B_{i_k} | z_{j_1} \cdots z_{j_m} } - \rho_{B_{i_1}|z_{j_1} \cdots z_{j_m} } \otimes \cdots \otimes \rho_{B_{i_k} | z_{j_1} \cdots z_{j_m} }  \|_1} \\
&\leq \sum_{t=1}^k \frac{1}{n-k+1} \sum_{m=0}^{n-k} \exc{\vec{i}, \vec{j}, z_{j_1}, \dots, z_{j_m} }{\| \rho_{B_{i_1 \cdots i_t} | z_{j_1} \cdots z_{j_m} } -  \rho_{B_{i_1 \cdots i_{t-1}} | z_{j_1} \cdots z_{j_m}  } \otimes \rho_{B_{i_{t}} | z_{j_1} \cdots z_{j_m} } \|_1} \\
&\leq kf(B|\cdot) \sqrt{(2 \ln 2) \frac{\log d_B^{k-1}}{n-k+1}}.
\end{align}
\end{proof}


\section{Constrained bilinear optimization}\label{sec:outer-bounds}

As stated in~\eqref{original_program}, the constrained bilinear optimization problem we are interested in takes the form
\begin{align}
Q:=\max &\quad \mathrm{Tr}\big[G_{AB} (W_{A} \otimes W_{B})\big]\\
s.t. &\quad W_{A} \succeq 0, W_{B} \succeq 0, \mathrm{Tr}(W_{A}) = \mathrm{Tr}(W_{B}) = 1\\
&\quad \Lambda_{A\to C_A}\left(W_{A}\right)=X_{C_A}, \: \Gamma_{B\to C_B}\left(W_{B}\right)=Y_{C_B}.
\end{align}
Lower bounds on the optimal value can, e.g., be derived by means of seesaw methods~\cite{Konno1976} (see~\cite{wolf01} for an example in quantum information theory). These then often converge in practice and sometimes even provably reach a local maxima. What was missing, however, is a general method to give an approximation guarantee to the global maximum.

Our de Finetti theorem with linear constraints (Theorem~\ref{thm:main-deFinetti}) gives an SDP hierarchy of outer bounds, that exactly provides such a criterion.

\begin{theorem}
For the SDPs
\begin{align}
\mathrm{SDP}_n :=\max &\quad \mathrm{Tr}\big[G_{AB} W_{AB_1}\big]\\
s.t. 
&\quad W_{AB_1^n}\succeq0, \mathrm{Tr}(W_{AB_1^n}) = 1, \;W_{AB_1^n}=\mathcal{U}_{B_1^n}^\pi\left(W_{AB_1^n}\right)\;\forall\pi\in\mathfrak{S}_n\\
& \quad \Lambda_{A\to C_A}\left(W_{AB_1^n}\right)=X_{C_A}\otimes W_{B_1^n},\;\Gamma_{B_n\to C_B}\left(W_{B_1^n}\right)=W_{B_1^{n-1}}\otimes Y_{C_B}
\end{align}
and $Q$ defined as above, we have for $d:=\max\{d_A,d_B\}$ that
\begin{align}
0 \leq \mathrm{SDP}_n - Q \leq \frac{\mathrm{poly}(d)}{\sqrt{n}}\quad\text{implying}\quad Q=\lim_{n\to\infty}\mathrm{SDP}_n.
\end{align}
\end{theorem}

\begin{proof}
We have by construction $0 \leq \mathrm{SDP}_n - Q $ and the remaining inequality arises from
\begin{align}
\mathrm{Tr}\left[G_{AB} W_{AB_1}\right]&=\mathrm{Tr}\left[G_{AB} (W_{A} \otimes W_{B})\right]+\mathrm{Tr}\left[G_{AB}\left(W_{AB_1} - W_{A} \otimes W_{B}\right)\right] \\
& \leq \mathrm{Tr}\left[G_{AB} (W_{A} \otimes W_{B})\right]+\| G_{AB} \|_\infty \cdot \|W_{AB_1} - W_{A} \otimes W_{B} \|_1 \\
&\leq\mathrm{Tr}\left[G_{AB} (W_{A} \otimes W_{B})\right]+\frac{\mathrm{poly}(d)}{\sqrt{n}},
\end{align}
where we used H\"older's inequality and the de Finetti argument as in Theorem \ref{thm:main-deFinetti}.
\end{proof}

The bounds from Theorem~\ref{thm:main-deFinetti} give worst case convergence guarantees that are slow\,---\,as to ensure that the approximation error is small we need at least the level $n=\mathrm{poly}(d)$. However, note that constrained bilinear optimization contains as a special case the best separable state problem and so we cannot expect much better bounds on the convergence speed in general. We refer to \cite{harrow2016limitations} and the references therein for a detailed discussion about the computational complexity of the best separable state problem.

We can add positive partial transpose (PPT) constraints\footnote{The partial transpose of a matrix $W_{AB}$ is defined for a fixed product basis as $\langle ij|W_{AB}^{T_A}|kl\rangle:=\langle kj|W_{AB}|il\rangle$.}
\begin{align}
W_{AB_1^n}^{T_A}\succeq0,\;W_{AB_1^n}^{T_{B_1}}\succeq0,\;W_{AB_1^n}^{T_{B_1^2}}\succeq0,\dots,\;W_{AB_1^n}^{T_{B_1^{n-1}}}\succeq0
\end{align}
to $\mathrm{SDP}_n$ and we denote the resulting relaxations by $\mathrm{SDP}_{n,\mathrm{PPT}}$. It is important to point out that any separable state is also a PPT state, and hence we still have a valid relaxation to the problem \eqref{original_program}. It is an interesting question to study if these constraints can lead to a faster convergence speed, cf.~the discussions in~\cite{Navascues09,Fawzi19}. Based on the PPT constraints, we can give a sufficient condition when already
\begin{align}
\text{$\mathrm{SDP}_{n,\mathrm{PPT}}=Q$ for some finite $n$.}
\end{align}
The condition\,---\,known as {\it rank loop condition}\,---\,is based on \cite{Navascues09}, which in turn builds on \cite{Horodecki00}.

\begin{lemma}{\cite{Navascues09},\cite{Horodecki00}}\label{lem:rank-loop}
Let $W_{AB_1^n}=\mathcal{U}_{B_1^n}^\pi\left(W_{AB_1^n}\right)$ for all $\pi\in\mathfrak{S}_n$ and fixed $0\leq k\leq n$ such that $W_{AB_1^n}^{T_{B_{k+1}^n}}\succeq0$. Then, $W_{AB_1}$ is separable if
\begin{align}
\mathrm{rank}(W_{AB_1^n}) \leq \max\left\{\mathrm{rank}\left(W_{AB_1^k}\right),\,\mathrm{rank}\left(W_{B_{k+1}^n}\right)\right\}.
\end{align}
\end{lemma}

Finally, note that instead of extending the $B$-systems we could equally well extend the $A$-systems to get another hierarchy. In the next section we directly study our main setting of interest\,---\,approximate quantum error correction\,---\,and refrain from further analysing the general case.


\section{Approximate quantum error correction}\label{sec:quantum-error-correction}

\subsection{Motivation}

In order to introduce the problem, we describe its relevance and applications in quantum information theory. First, we introduce the theoretical setting, then we apply the results of the previous sections, thus obtaining specific convergent hierarchies. Corresponding numerical tests can be found in Appendix~\ref{sec:numerics}.

Given a noisy classical channel $N_{X\to Y}$, a central quantity of interest in error correction is the maximum success probability $p(N,M)$ for transmitting a uniform $M$-dimensional message under the noise model $N_{X\to Y}$. This is a bilinear maximization problem, which is in general NP-hard to approximate up to a sufficiently small constant factor~\cite{fawzi18}. Nevertheless, there are efficient methods for constructing feasible coding schemes approximating $p(N,M)$ from below as well as an efficiently computable linear programming relaxation $\mathrm{lp}(N,M)$ (sometimes called meta converse~\cite{hayashi09,polyanskiy10}) giving upper bounds on $p(N,M)$.\footnote{Operationally, $\mathrm{lp}(N,M)$ corresponds to the non-signalling assisted maximum success probability~\cite{matthews11}.} In fact, it was shown in~\cite{fawzi18} that $p(N,M)$ and $\mathrm{lp}(N,M)$ cannot be very far from each other
\begin{align}
p(N,M)\leq\mathrm{lp}(N,M)\leq\frac{1}{1-\frac{1}{e}}\cdot p(N,M).
\end{align}
Furthermore, the meta-converse has many appealing analytic properties, such as, e.g., the ability to evaluate it efficiently in the limit of many independent repetitions $N_{X\to Y}^{\times n}$, leading to very precise asymptotic bounds on the capacity of noisy classical channels \cite{fawzi18}.

The analogue quantum problem is to determine the maximum fidelity $F(\mathcal{N},M)$, a quantity that will be formally defined later (Definition \ref{def:plain-coding}), for transmitting one part of a maximally entangled state of dimension $M$ over a noisy quantum channel $\mathcal{N}_{A\to B}$. As in the classical case, this is a bilinear optimization problem, only now with matrix-valued variables. In order to approximate $F(\mathcal{N},M)$, an efficiently computable semidefinite programming relaxation $\mathrm{SDP}(\mathcal{N},M)$ was given in~\cite{matthews14}.\footnote{Operationally, $\mathrm{SDP}(N,M)$ corresponds to the positive partial transpose preserving, non-signalling assisted maximum fidelity~\cite{matthews14}.} However, contrary to the classical case the gap between $\mathrm{SDP}(\mathcal{N},M)$ and $F(\mathcal{N},M)$ is not understood. On the other hand, the tools introduced in Section \ref{sec:deFinetti} will exactly be used to generate a converging hierarchy of efficiently computable semidefinite programming relaxations, allowing us to quantify the gap between these new relaxations and $F(\mathcal{N},M)$.

Moreover, the relaxation $\mathrm{SDP}(\mathcal{N},M)$ is lacking most of the analytic properties of its classical analogue $\mathrm{lp}(N,M)$. In fact, in quantum communication theory so-called non-additivity problems caused by quantum correlations make it notoriously hard to compute asymptotic limits in the first place~\cite{divincenzo98}. Hence, we propose to use methods from optimization theory to directly study the maximum fidelity $F(\mathcal{N},M)$ in order to quantify the ability of a quantum channel to transmit quantum information. The goal is then to identify a quantum version of the meta converse for approximating $F(\mathcal{N},M)$, having similar properties as the classical meta converse $\mathrm{lp}(N,M)$ for approximating $p(N,M)$. This approach can also be justified by the fact that most of the quantum devices that will be available in the near future are likely to be noisy and small in size. As such, efficient algorithms approximating $F(\mathcal{N},M)$ for reasonable error models $\mathcal{N}$ and dimension $M$ are more relevant in such settings than computing the asymptotic limit of the rate achievable for multiple copies of a given noise model.

Numerical lower bound methods for $F(\mathcal{N},M)$ are available through iterative seesaw methods that lead to efficiently computable semidefinite programs~\cite{Werner05,reimpell08,Shor07,fletcher08,Kosut2009,lidar10,aspuru-guzik17}. These algorithms often converge in practice and sometimes even provably reach a local maximum. What was previously missing, however, is a general method to give an approximation guarantee to the global maximum. Here, the techniques as developed in Section~\ref{sec:outer-bounds} exactly lead to a converging hierarchy of efficiently computable semidefinite programming relaxations on the maximum fidelity $F(\mathcal{N},M)$. As such, this can be seen as a tool for benchmarking existing quantum error correction codes and to understand in what direction to look for improved codes

We note that references \cite{Tomamichel:2016aa,wang16,wang17,kaur18} gave refined relaxations on the size of a maximally entangled state that can be sent over a noisy quantum channel for fixed fidelity $1-\eps$. These approaches are complementary to our work and contrary to our findings they do not lead to a converging hierarchy of efficiently computable bounds.


\subsection{Setting}

The mathematical setting of approximate quantum error correction we study is as follows.

\begin{definition}\label{def:plain-coding}
	Let $\mathcal{N}_{\bar A\to B}$ be a quantum channel and $M \in \mathbb{N}$. The channel fidelity for message dimension $M$ is defined as
	\begin{align}
	F(\mathcal{N},M):=\max &\quad F\Big(\Phi_{\bar BR},\big(\left(\mathcal{D}_{B\to\bar B}\circ\mathcal{N}_{\bar A\to B}\circ\mathcal{E}_{A\to\bar A}\right)\otimes\mathcal{I}_R\big)(\Phi_{AR})\Big)\\
	s.t. &\quad \mathcal{D}_{B\to\bar B},\mathcal{E}_{A\to\bar A}\;\text{quantum channels},
	\end{align}
	where $F(\rho,\sigma):=\left\|\sqrt{\rho}\sqrt{\sigma}\right\|_1^2$ denotes the fidelity, $\Phi_{ A R}$ denotes the maximally entangled state on $AR$, and we have $M=d_A=d_{\bar B}=d_R$.
\end{definition}

In information-theoretic language, the channel fidelity corresponds to an average error criteria for preserving uniformly distributed information. Alternatively, we might also aim for a worst error criteria and we discuss this in Appendix~\ref{sec:diamondNorm}.

By the Choi-Jamio\l{}kowski isomorphism the channel fidelity is conveniently rewritten as a bilinear optimization.

\begin{lemma}\label{lem:Choi}
	Let $\mathcal{N}_{\bar A\to B}$ be a quantum channel and $M\in\mathbb{N}$. Then, the channel fidelity can be written as
	\begin{align}
	F(\mathcal{N},M)=\max &\quad d_{\bar A}d_B\cdot\mathrm{Tr}\left[\left(J^\mathcal{N}_{\bar AB}\otimes\Phi_{A\bar B}\right)\left(E_{A\bar A}\otimes D_{B\bar B}\right)\right] \\
	s.t. &\quad E_{A\bar A}\succeq0,\;D_{B\bar B}\succeq0,\;E_A=\frac{1_A}{d_A},\;D_B=\frac{1_B}{d_B},
	\end{align}
	where $J^\mathcal{N}_{B\bar A}:=(\mathcal{N}_{\bar A \to B}\otimes\mathcal{I}_{\bar A})(\Phi_{\bar A\bar A})$ denotes the Choi state of $\mathcal{N}_{\bar A\to B}$.
\end{lemma}

The advantage of this notation is that all $A$-systems are with the sender (termed Alice) and all $B$-systems are with the receiver (termed Bob), which is consistent with~\cite{matthews14}.

\begin{proof}
By using the adjoint map in Hilbert-Schmidt inner product and multiple times the Choi-Jamio\l{}kowski isomorphism as given in \eqref{eq:Choi-state} and \eqref{eq:Choi-channel}, and noting that $\Phi_{\bar BR}$ allows us to use the simplified expression for the fidelity when one of the two arguments is pure \cite[Section 9.2]{Wilde-Book}, we can write the objective function from Definition~\ref{def:plain-coding} as
\begin{align}
&F\Big(\Phi_{\bar BR},\big(\left(\mathcal{D}_{B\to\bar B}\circ\mathcal{N}_{\bar A\to B}\circ\mathcal{E}_{A\to\bar A}\right)\otimes\mathcal{I}_R\big)(\Phi_{AR})\Big)\\
&=\tr\left[\Phi_{\bar BR}\big(\left(\mathcal{D}_{B\to\bar B}\circ\mathcal{N}_{\bar A\to B}\circ\mathcal{E}_{A\to\bar A}\right)\otimes\mathcal{I}_R\big)(\Phi_{AR})\right]\\
&= \tr\left[J^\mathcal{D^\dagger}_{BR} \left(\mathcal{N}_{\bar A\to B} \otimes\mathcal{I}_R\right)\left(J^\mathcal{E}_{\bar A R}\right)\right].
\end{align}
Taking advantage of $d_A=d_{\bar B}=d_R$, we relabel the systems and we proceed as follows
\begin{align}
&F\Big(\Phi_{\bar BR},\big(\left(\mathcal{D}_{B\to\bar B}\circ\mathcal{N}_{\bar A\to B}\circ\mathcal{E}_{A\to\bar A}\right)\otimes\mathcal{I}_R\big)(\Phi_{AR})\Big)\\
&= \tr\left[J^\mathcal{D^\dagger}_{BR} \left(\mathcal{N}_{\bar A\to B} \otimes\mathcal{I}_R\right)\left(J^\mathcal{E}_{\bar A R}\right)\right]\\		
&= \tr\left[J^\mathcal{D^\dagger}_{B\bar B} \left(\mathcal{N}_{\bar A\to B} \otimes\mathcal{I}_{A \to \bar B} \right)\left(J^\mathcal{E}_{\bar A A}\right)\right] \\
&= d_A d_{\bar A} \cdot \tr \left[ \left(J^\mathcal{N}_{\bar AB} \otimes \Phi_{A\bar B} \right)\left(\left(J^\mathcal{E}_{A\bar A}\right)^T \otimes J^\mathcal{D^\dagger}_{B\bar B} \right) \right] \\
&=  d_{\bar A} d_B \cdot \tr \left[\left(J^\mathcal{N}_{\bar AB} \otimes \Phi_{A\bar B} \right)\left(\left(J^\mathcal{E}_{A\bar A} \right)^T \otimes\frac{d_A}{d_B}\cdot J^\mathcal{D^\dagger}_{B\bar B} \right) \right],
\end{align}
where the transpose is taken with respect to the canonical basis, and the dimensional factors come from the connection between the Hilbert-Schmidt inner product and the maximally entangled state \cite[Example 1.2]{Wolf-Notes}. Due to the basic proprieties of the Choi-Jamio\l{}kowski isomorphism it is immediate to see that $(J^\mathcal{E}_{A\bar A})^T$ can be identified with the $E_{A\bar A}$ of Lemma~\ref{lem:Choi}. In addition, we have $\frac{d_A}{d_B}\cdot J^\mathcal{D^\dagger}_{B\bar B}\succeq0$, and tracing out the $\bar B$ system as well as using $d_A = d_{\bar B}$ we get
\begin{align}
\frac{d_A}{d_B}\cdot J^\mathcal{D^\dagger}_{B} = \frac{d_A}{d_B}\cdot \mathcal{D}^\dagger\left(\frac{1_{\bar B}}{d_{\bar B}}\right) = \frac{d_A}{d_B}\cdot\frac{1}{d_{\bar B}}\cdot1_{B} = \frac{1_B}{d_B}.
\end{align}
Thus, we can identify $\frac{d_A}{d_B}\cdot J^\mathcal{D^\dagger}_{B\bar B}$ with the $D_{B\bar B}$ of Lemma~\ref{lem:Choi}. 
\end{proof}

The following simple dimension bounds hold for the channel fidelity.

\begin{restatable}[]{lemma}{fdim}
	\label{lem:F-dimension}
	Let $\mathcal{N}_{\bar A\to B}$ be a quantum channel and $M\in\mathbb{N}$. Then, we have
	\begin{align}
	0\leq F(\mathcal{N},M)\leq\min\left\{1,\left(\frac{d_{\bar A}}{M}\right)^2,\frac{d_B}{M}\right\}.
	\end{align}
\end{restatable}

The proof can be found in Appendix~\ref{app:missing-proofs}. By the linearity of the objective function we can furthermore rewrite the channel fidelity as
\begin{align}
F(\mathcal{N},M)=\max &\quad d_{\bar A}d_B\cdot\mathrm{Tr}\left[\left(J^\mathcal{N}_{\bar AB}\otimes\Phi_{A\bar B}\right)\left(\sum_{i\in I}p_iE_{A\bar A}^i\otimes D_{B\bar B}^i\right)\right] \\
s.t. &\quad p_i\geq0\;\forall i\in I,\;\sum_{i\in I}p_i=1\\
& \quad E_{A\bar A}^i\succeq0,\;D_{B\bar B}^i\succeq0,\;E^i_A=\frac{1_A}{d_A},\;D^i_B=\frac{1_B}{d_B}\;\forall i\in I.
\end{align}


\subsection{De Finetti theorems for quantum channels}\label{sec:separable-deFinetti}

Recall that a quantum channel is just a completely positive, trace preserving map between two spaces of quantum states. Here, we establish a sufficient criterion under which permutation invariance of a quantum channel implies that it can be well approximated by a mixture of product quantum channels.

\begin{theorem}\label{lem:sep-deFinetti}
Let $\rho_{A\bar A(B\bar B)_1^n}$ be a quantum state with
\begin{align}\label{eq:constraint-perm-inv}
\rho_{A\bar A(B\bar B)_1^n}&=\mathcal{U}_{(B\bar B)_1^n}^\pi(\rho_{A\bar A(B\bar B)_1^n})\;\forall\pi\in\mathfrak{S}_n\\
\label{eq:constraint-tra}
\rho_{A(B\bar B)_1^n}&=\frac{1_A}{d_A}\otimes\rho_{(B\bar B)_1^n}\\
\label{eq:constraint-trb}
\rho_{(B\bar B)_1^{n-1}B_n}&=\rho_{(B\bar B)_1^{n-1}}\otimes\frac{1_{B_n}}{d_B}.
\end{align}
Then, we have for $0<k<n$ that
\begin{align}
\left\|\rho_{A\bar A(B\bar B)_1^{k}}-\sum_{i\in I}p_i\sigma_{A\bar A}^i\otimes\left(\omega_{B\bar B}^i\right)^{\otimes k}\right\|_1 &\leq kf(B \bar{B} |\cdot)\sqrt{ (2 \ln 2) \frac{ \log(d_{A} d_{\bar{A}}) + (k - 1) \log(d_B d_{\bar{B}}) }{n-k+1}}
\end{align}
with $\{p_i\}_{i\in I}$ a probability distribution, and $\sigma_{A\bar A}^i,\omega_{B\bar B}^i\succeq0$ such that $\sigma_A^i=\frac{1_A}{d_A}$ and $\omega_B^i=\frac{1_B}{d_B}$ for $i\in I$.
\end{theorem}

\begin{proof}
We simply apply Theorem~\ref{cor:main-deFinetti} for the linear maps $\Lambda_{A\bar{A} \to A} = \tr_{\bar{A}}$ and $\Gamma_{B\bar{B} \to B} = \tr_{\bar{B}}$.
\end{proof}

We emphasize that the representation we obtain in this theorem, $\rho_{A\bar A(B\bar B)_1^{k}}$ is close to a mixture of products of Choi states of completely positive and \emph{trace-preserving} maps. We note that applying standard de Finetti theorems for quantum states would only show that $\rho_{A\bar A(B\bar B)_1^{k}}$ is close to a mixture of products states\,---\,or in other words Choi states of completely positive maps that are in general not even trace-non-increasing. This is not sufficient for our applications, and having the constraints~\eqref{eq:constraint-trb} and~\eqref{eq:constraint-tra} are needed in our proofs to achieve this stronger statement. We discuss this in more detail by means of the following examples.

\begin{example}
For $\bar A\bar B$ trivial and $k=1$ Theorem~\ref{lem:sep-deFinetti} says that $\rho_{AB}$ is close to the product state $\frac{1_{AB}}{d_Ad_B}$, as this is the only valid state satisfying the linear constraints. However, having only the permutation invariance condition~\eqref{eq:constraint-perm-inv} without the other two conditions in Theorem~\ref{lem:sep-deFinetti}, this conclusion does not hold. In fact, choose $\rho_{AB_1^n}$ to be maximally classically correlated between all systems $A;B_1;B_2^n$
\begin{align}
\rho_{AB_1^n} = \frac{1}{d} \sum_{i} \ket{i}\bra{i}^{\otimes n+1}.
\end{align}
Then, the systems $B_1^n$ are symmetric with respect to $A$ and even more, the state is supported on the symmetric subspace $(1_{A} \otimes P^{\mathrm{sym}}_{B_1^n})(\rho_{A B_1^n}) = \rho_{A B_1^n}$. However, of course $\rho_{AB_1}$ is not close to the state $\frac{1_{AB_1}}{d_Ad_B}$.
\end{example}

\begin{example}\label{example:weak-condition}
This following example shows that imposing the constraint $\rho_{AB_1} = \frac{1_{AB_1}}{d_Ad_B}$ is not enough either. Let $A,\bar A, B, \bar B$ all be of dimension $d \geq 2$. Then, define for any $n \geq 1$
\begin{align}
\rho_{AB_1^n \bar{A} \bar{B}_1^n} = \frac{1}{d^2} \sum_{i,j}  \ket{j}\bra{j}_{A} \otimes \ket{i}\bra{i}_{\bar A} \otimes \ket{i}\bra{i}^{\otimes n}_{B} \otimes \ket{i}\bra{i}^{\otimes n}_{\bar B}.
\end{align}
Then, the state is invariant under permutations of the $B\bar{B}$ systems and $\rho_{AB_1} =  \frac{1_{AB_1}}{d^2}$. However, the reduced state $\rho_{A\bar{A} B_1 \bar{B}_1}$ is not close to states of the form
\begin{align}
\sum_{\ell} p_{\ell} \sigma^{\ell}_{A\bar{A}} \otimes \omega^{\ell}_{B \bar{B}}\quad\text{with}\quad\sigma^{\ell}_{A} =\frac{1_A}{d}, \omega^{\ell}_{B} = \frac{1_B}{d}.
\end{align}
To see this, consider the projector $\Pi_{\bar{A} B} = \sum_{i} \ket{i}\bra{i}_{\bar{A}}\otimes \ket{i}\bra{i}_{B}$. Then, we get 
\begin{align}
\mathrm{Tr}(\Pi_{\bar{A} B} \rho_{A\bar{A} B\bar{B}}) = 1\quad\text{but}\quad\mathrm{Tr}(\Pi_{\bar{A} B} \sigma^{\ell}_{A\bar{A}} \otimes  \omega^{\ell}_{B\bar{B}}) = \mathrm{Tr}(\Pi_{\bar{A} B} \sigma^{\ell}_{\bar{A}} \otimes \frac{1_B}{d}) = \frac{1}{d}.
\end{align}
\end{example}

By the Choi-Jamio\l{}kowski isomorphism and relating the trace norm distance of Choi states to the diamond norm distance of the quantum channels~\cite[Lemma 7]{Wallman14}, we can alternatively state the bounds from Theorem~\ref{lem:sep-deFinetti} directly in terms of the quantum channels.

\begin{corollary}\label{cor:exchangeable-deFinetti}
Let $\mathcal{N}_{AB_1^{n}\to\bar A\bar B_1^{n}}$ be a quantum channel such that
\begin{align}
\mathcal{U}_{\bar B_1^n}^\pi\left(\mathcal{N}_{AB_1^n\to\bar A\bar B_1^n}(\cdot)\right)&=\mathcal{N}_{AB_1^n\to\bar A\bar B_1^n}\left(\mathcal{U}_{B_1^n}^\pi(\cdot)\right)\;\forall\pi\in\mathfrak{S}_n\label{eq:sym-condition}\\
\mathrm{Tr}_{\bar B_n}\Big[\mathcal{N}_{AB_1^n\to\bar A\bar B_1^n}(\cdot)\Big] &= \mathrm{Tr}_{\bar B_n}\left[\mathcal{N}_{AB_1^n\to\bar A\bar B_1^n}\left(\mathrm{Tr}_{B_n}[\cdot] \otimes \frac{1_{B_n}}{d_B} \right) \right] \label{eq:partial-trace-bn}\\
\mathrm{Tr}_{\bar A}\left[\mathcal{N}_{AB_1^n\to\bar A\bar B_1^n}(\cdot)\right]&=\mathrm{Tr}_{\bar A}\left[\mathcal{N}_{AB_1^n\to\bar A\bar B_1^n}\left(\frac{1_A}{d_A}\otimes\mathrm{Tr}_A\left[\cdot\right]\right)\right].\label{eq:partial-trace-an}
\end{align}
Then, we have for $0<k<n$ with
\begin{align}\label{eq:reduction}
\mathcal{N}_{AB_1^{k}\to\bar A\bar B_1^{k}}(X_{A B_1^k}):= \mathrm{Tr}_{\bar{B}_{k+1}^n}\left[\mathcal{N}_{AB_1^n\to\bar A\bar B_1^n}\left(X_{A B_1^k} \otimes \frac{1_{B_{k+1}^n}}{d_B^{n-k}}\right) \right]
\end{align}
that
\begin{align}
\left\|\mathcal{N}_{AB_1^k \to\bar A\bar B_1^k}-\sum_{i\in I}p_i\mathcal{E}_{A\to\bar A}^i\otimes\left(\mathcal{D}_{B\to\bar B}^i\right)^{\otimes k}\right\|_\Diamond\leq\;&d_Ad_B^k\cdot kf(B \bar{B}|\cdot)  \\
& \times\sqrt{ (2 \ln 2) \frac{ \log(d_{A} d_{\bar{A}}) + (k - 1) \log(d_B d_{\bar{B}}) }{n-k+1}}
\end{align}
with $\{p_i\}_{i\in I}$ a probability distribution and $\mathcal{D}_{B\to\bar B}^i,\mathcal{E}_{A\to\bar A}^i$ quantum channels for $i\in I$.
\end{corollary}

In~\eqref{eq:reduction} we chose a specific extension of $X_{AB_1^k}$ to define $\mathcal{N}_{AB_1^{k}\to\bar A\bar B_1^{k}}(X_{A B_1^k})$, namely $X_{AB_1^k} \otimes \frac{1_{B_{k+1}^n}}{d_B^{n-k}}$. This is still well-defined as the conditions \eqref{eq:sym-condition} and \eqref{eq:partial-trace-bn} we require of $\mathcal{N}_{AB_1^n\to\bar A\bar B_1^n}$ actually say that the choice of extension does not matter. That is, we have for any $X_{A B_1^n}$ that
\begin{align}
\mathrm{Tr}_{\bar{B}_{k+1}^n} \left[\mathcal{N}_{AB_1^{n}\to\bar A\bar B_1^{n}}(X_{A B_1^n}) \right] 
&= \mathrm{Tr}_{\bar{B}_{k+1}^{n-1} }\left[ \mathrm{Tr}_{\bar{B}_{n} }\left[\mathcal{N}_{AB_1^n\to\bar A\bar B_1^n}\left(X_{A B_1^{n-1}} \otimes \frac{1_{B_n}}{d_B}\right) \right] \right]\\
&= \mathrm{Tr}_{\bar{B}_{k+1}^{n} }\left[\mathcal{N}_{AB_1^n\to\bar A\bar B_1^n}\left(X_{A B_1^{k}} \otimes \frac{1_{B_{k+1}^n}}{d_B^{n-k}}\right) \right] \\
&= \mathcal{N}_{AB_1^k\to\bar A\bar B_1^k}\left(X_{A B_1^{k}}\right) \,,
\end{align}
where we used \eqref{eq:partial-trace-bn} for the first equality as well as \eqref{eq:sym-condition} and \eqref{eq:partial-trace-bn} multiple times for the second equality.

In the following we state several comments about de Finetti theorems for quantum channels:
\begin{itemize}
\item We emphasize that the de Finetti reductions\,---\,called post-selection technique~\cite{christandl09}\,---\,for quantum channels proved in~\cite{fawzirenner14,brandao14} are different from what we need in our work (also see \cite{Rotem15,HT14} for classical versions). In particular, unlike de Finetti theorems, de Finetti reductions provide an operator inequality upper bound to a symmetric quantum state in the form of an integral superposition of product states.
\item In contrast to the bound for Choi states (Theorem~\ref{lem:sep-deFinetti}), the diamond norm bound in Corollary~\ref{cor:exchangeable-deFinetti} does not have a polynomial dependence in $d_B$ and $k$. We leave it as an open question to give a de Finetti theorem for quantum channels in terms of the diamond norm distance with a dimension dependence polynomial in $d_B$ and $k$. (For our purposes we only need the $k=1$ bound, in terms of the Choi states.)
\item In the case $k=1$, the conditions of the above theorem can be seen as approximations for the convex hull of product quantum channels, just as extendible states provide an approximation for the set of separable states.\footnote{The class of channels we consider here is more restricted than general separable channels, which usually refers to a mixture of product completely positive and not necessarily trace-preserving maps.} We note that in SDP hierarchies for the quantum separability problem the permutation invariance can be replaced by the stronger {\it Bose symmetric} condition. That is, the state in question is supported on the symmetric subspace. The reason is that every separable state can without loss of generality be decomposed in a convex combination of {\it pure} product states. However, in our setting, we cannot assume that we have a mixture of a product of pure channels, and so we keep the more general notion of permutation invariance.
\item In the following, we never directly make use of Corollary~\ref{cor:exchangeable-deFinetti} but rather state it for connecting to the previous literature. In particular, when choosing $A\bar A$ trivial as a special case we find a finite version of the asymptotic de Finetti theorem for quantum channels from~\cite{Fuchs04,Fuchs2004}.\footnote{We also refer to~\cite{Pankowski13} for previous related work and \cite{toner09} for a classical version. Moreover, following~\cite{kaur18}, conditions related to our~\eqref{eq:sym-condition} -- \eqref{eq:partial-trace-an} give rise to extendible channels in the resource theory of unextendibility.} We emphasize that our derived conditions then become a finite version of the notion of {\it exchangeable sequences of quantum channels} of~\cite{Fuchs04} defined as a sequence of channels $\{\mathcal{N}_{B_1^n\to\bar B_1^n}\}$ satisfying for all $n$ that
\begin{align}
\mathcal{U}_{\bar B_1^n}^\pi\left(\mathcal{N}_{B_1^n\to\bar B_1^n}(\cdot)\right)&=\mathcal{N}_{B_1^n\to\bar B_1^n}\left(\mathcal{U}_{B_1^n}^\pi(\cdot)\right)\;\forall\pi\in\mathfrak{S}_n\\
\mathcal{N}_{B_1^{n-1}\to\bar B_1^{n-1}}\Big(\mathrm{Tr}_{B_n}\left[\cdot\right]\Big)&=\mathrm{Tr}_{\bar B_n}\Big[\mathcal{N}_{B_1^n\to\bar B_1^n}(\cdot)\Big].
\end{align}
They show that under these conditions, for any $k$, the channel $\cN_{B_1^k \to \bar{B}_1^k}$ is in the convex hull of tensor power channels. In Corollary~\ref{cor:exchangeable-deFinetti}, we start with a channel\footnote{This is equivalent to being given a finite sequence $\mathcal{N}_{B_1^k\to\bar B_1^k}$ for $k \in \{1, \dots, n\}$ satisfying the exchangeability condition, as the reduced channels are then completely determined by $\mathcal{N}_{B_1^n\to\bar B_1^n}$} $\mathcal{N}_{B_1^n\to\bar B_1^n}$ and quantify the closeness of such $\mathcal{N}_{B_1^k\to\bar B_1^k}$ to convex combinations of tensor product channels $\sum_ip_i\left(\mathcal{D}_{B\to\bar B}^i\right)^{\otimes k}$.
\end{itemize}

Channels that are written as mixtures of channels of the form $\mathcal{E}_{A \to \bar A} \otimes \mathcal{D}_{B \to \bar B}$ where $\mathcal{E}_{A \to \bar A}$ and $\mathcal{D}_{B \to \bar B}$ are channels can straightforwardly be implemented between two parties having access to shared randomness but no communication. There is a natural relaxation to this set of channels, often called LOCC(1) channels, corresponding to channels that can be implemented with additional classical communication from $A$ to $B$. Mathematically, these are channels of the form
\begin{align}
\sum_{i \in I}\mathcal{E}^i_{A \to \bar A} \otimes \mathcal{D}^i_{B \to \bar B},
\end{align}
where $\mathcal{D}^i_{B\to\bar B}$ are channels and $\mathcal{E}^i_{A\to\bar A}$ are completely positive but not necessarily trace-preserving. We discuss this variation of approximate quantum error correction in Appendix~\ref{sec:appendix}.


\subsection{Hierarchy of outer bounds}\label{subsectionHierarchy}

Following the de Finetti theorem for quantum channels as given in Theorem~\ref{lem:sep-deFinetti}, the $n$-th level of the SDP hierarchy for the quantum channel fidelity becomes
\begin{align}
\mathrm{SDP}_n(\mathcal{N},M):=\max &\quad d_{\bar A}d_B\cdot\mathrm{Tr}\left[\left(J^\mathcal{N}_{\bar AB_1}\otimes\Phi_{A\bar B_1}\right)W_{A\bar AB_1\bar B_1}\right]\\
s.t. &\quad W_{A\bar A(B\bar B)_1^n}\succeq0,\;\mathrm{Tr}\left[W_{A\bar A(B\bar B)_1^n}\right]=1\\
& \quad W_{A\bar A(B\bar B)_1^n}=\mathcal{U}_{(B\bar B)_1^n}^\pi\left(W_{A\bar A(B\bar B)_1^n}\right)\;\forall\pi\in\mathfrak{S}_n\\
& \quad W_{A(B\bar B)_1^n}=\frac{1_A}{d_A}\otimes W_{(B\bar B)_1^n},\;W_{A\bar A(B\bar B)_1^{n-1}B_n}=W_{A\bar A(B\bar B)_1^{n-1}}\otimes\frac{1_{B_n}}{d_B}.
\end{align}
Here, we identified $B_1\equiv B$ and hence the $n$-th level of the hierarchy then corresponds to taking $n-1$ extensions. Note that instead of stating the last condition for the final block $B_n$ we could have equivalently stated it for any block $B_j$ with $j=1,\ldots,n$ (by the permutation invariance). Iteratively, the condition then also holds on all pairs of blocks of size two, and so on. Moreover, we slightly strengthened the last condition by including the $A$-systems compared to the minimal condition on the $B$-system needed for Theorem~\ref{lem:sep-deFinetti}
\begin{align}
W_{(B\bar B)_1^{n-1}B_n}=W_{(B\bar B)_1^{n-1}}\otimes\frac{1_{B_n}}{d_B}.
\end{align}
We then immediately have asymptotic convergence.

\begin{theorem}\label{thm:convergence-plain}
Let $\mathcal{N}$ be a quantum channel and $n,M\in\mathbb{N}$. Then, we have
\begin{align}
0 \leq \mathrm{SDP}_n(\mathcal{N},M) - F(\mathcal{N},M) \leq \frac{\mathrm{poly}(d)}{\sqrt{n}}\quad\text{implying}\quad F(\mathcal{N},M)=\lim_{n\to\infty}\mathrm{SDP}_n(\mathcal{N},M),
\end{align}
where $d:=\max\{d_A,d_{\bar A},d_B,d_{\bar B}\}$.
\end{theorem}

\begin{proof}
By construction $0 \leq \mathrm{SDP}_n(\mathcal{N},M) - F(\mathcal{N},M)$ and the remaining inequality arises from
\begin{align}
d_{\bar A}d_B\cdot\mathrm{Tr}\left[\left(J^\mathcal{N}_{\bar AB}\otimes\Phi_{A\bar B}\right)W_{A\bar A B \bar B}\right]&=  d_{\bar A}d_B\cdot\mathrm{Tr}\left[\left(J^\mathcal{N}_{\bar AB}\otimes\Phi_{A\bar B}\right)\left(E_{A\bar A}\otimes D_{B\bar B}\right)\right] \\
&\quad+ d_{\bar A}d_B\cdot\mathrm{Tr}\left[\left(J^\mathcal{N}_{\bar AB}\otimes\Phi_{A\bar B}\right)\left(W_{A\bar A B \bar B} - E_{A\bar A}\otimes D_{B\bar B}\right)\right] \\
& \leq d_{\bar A}d_B\cdot\mathrm{Tr}\left[\left(J^\mathcal{N}_{\bar AB}\otimes\Phi_{A\bar B}\right)\left(E_{A\bar A}\otimes D_{B\bar B}\right)\right] \\
&\quad + d_{\bar A}d_B \cdot \| J^\mathcal{N}_{\bar AB}\otimes\Phi_{A\bar B} \|_\infty \cdot \|W_{A\bar A B \bar B} - E_{A\bar A}\otimes D_{B\bar B} \|_1 \\
&\leq d_{\bar A}d_B\cdot\mathrm{Tr}\left[\left(J^\mathcal{N}_{\bar AB}\otimes\Phi_{A\bar B}\right)\left(E_{A\bar A}\otimes D_{B\bar B}\right)\right]+\frac{\mathrm{poly}(d)}{\sqrt{n}},
\end{align}
where we used H\"older's inequality and the de Finetti reduction from Theorem \ref{thm:main-deFinetti}.
\end{proof}

We note that the worst case convergence guarantee is slow, as to ensure that the approximation error becomes small, we need at least the level $n=\mathrm{poly}(d)$.

\begin{remark}
Instead of extending the $B$-systems we could alternatively extend the $A$-systems, which leads to the (non-equivalent) hierarchy
\begin{align}
\overline{\mathrm{SDP}}_n(\mathcal{N},M):=\max &\quad d_{\bar A}d_B\cdot\mathrm{Tr}\left[\left(J^\mathcal{N}_{\bar AB_1}\otimes\Phi_{A\bar B_1}\right)W_{A_1\bar A_1B\bar B}\right]\\
s.t. &\quad W_{(A\bar A)_1^nB\bar B}\succeq0,\;\mathrm{Tr}\left[W_{(A\bar A)_1^nB\bar B}\right]=1\\
& \quad W_{(A\bar A)_1^nB\bar B}=\mathcal{U}_{(A\bar A)_1^n}^\pi\left(W_{(A\bar A)_1^nB\bar B}\right)\;\forall\pi\in\mathfrak{S}_n\\
& \quad W_{(A\bar A)_1^nB}=W_{(A\bar A)_1^n}\otimes\frac{1_B}{d_B},\;W_{(A\bar A)_1^{n-1}A_nB\bar B}=\frac{1_{A_n}}{d_A}\otimes W_{(A\bar A)_1^{n-1}B\bar B}.
\end{align}
For the first level we have $\overline{\mathrm{SDP}}_1(\mathcal{N},M)=\mathrm{SDP}_1(\mathcal{N},M)$ by inspection, but for the higher levels it depends on the input-output dimensions $d_{\bar A},d_B$ which hierarchy is potentially more powerful.
\end{remark}

The relaxations $\mathrm{SDP}_n(\mathcal{N},M)$ behave naturally with respect to the first two bounds of Lemma~\ref{lem:F-dimension}.

\begin{restatable}[]{lemma}{boundsdpn}
Let $\mathcal{N}_{\bar A\to B}$ be a quantum channel and $n,M \geq 1$. Then, we have
\begin{align}
0\leq\mathrm{SDP}_n(\mathcal{N},M)\leq\min\left\{1,\left(\frac{d_{\bar A}}{M}\right)^2\right\}.
\end{align}
\end{restatable}

The proof can be found in Appendix~\ref{app:missing-proofs}. We can again add all the PPT constraints and denote the resulting relaxations by $\mathrm{SDP}_{n,\mathrm{PPT}}(\mathcal{N},M)$. In the following we study more closely these levels $\mathrm{SDP}_{n,\mathrm{PPT}}(\mathcal{N},M)$, which are our tightest outer bound relaxations on the channel fidelity.


\subsection{Low level relaxations}\label{sec:higherlevels-plain}

We find
\begin{align}
\mathrm{SDP}_{1,\mathrm{PPT}}(\mathcal{N},M)=\max &\quad d_{\bar A}d_B\cdot\mathrm{Tr}\left[\left(J^\mathcal{N}_{\bar AB}\otimes\Phi_{A\bar B}\right)W_{A\bar AB\bar B}\right]\\
s.t. &\quad W_{A\bar AB\bar B}\succeq0,\;W_{A\bar AB\bar B}^{T_{B\bar B}}\succeq0,\;\mathrm{Tr}\left[W_{A\bar AB\bar B}\right]=1\\
& \quad W_{AB\bar B}=\frac{1_A}{d_A}\otimes W_{B\bar B},\;W_{A\bar AB}=W_{A\bar A}\otimes\frac{1_B}{d_B},
\end{align}
which is the SDP outer bound found in~\cite[Section IV]{matthews14}, up to their a priori stronger condition
\begin{align}
\text{$W_{AB}=\frac{1_{AB}}{d_Ad_B}$ instead of our $\mathrm{Tr}\left[W_{A\bar AB\bar B}\right]=1$.}
\end{align}
However, as implicitly shown in~\cite[Theorem 3]{matthews14} these two conditions actually become equivalent because of the structure of the objective function. Operationally $\mathrm{SDP}_1(\mathcal{N},M)$ corresponds to the non-signalling assisted channel fidelity, whereas $\mathrm{SDP}_{1,\mathrm{PPT}}(\mathcal{N},M)$ adds the PPT-preserving constraint\,---\,as discussed in~\cite[Corollary 4]{matthews14}. Moreover, in the objective function the symmetry\footnote{Here, $\overline{U}_A$ denotes the complex conjugate of $U_A$ with respect to some standard basis.}
\begin{align}
\int\left(\overline U_A\otimes U_{\bar B}\right)(\cdot)\left(\overline U_A\otimes U_{\bar B}\right)^\dagger\,\mathrm{d}U
\end{align}
can be used to achieve a dimension reduction of $M^2$ leading to~\cite[Theorem 3]{matthews14}
\begin{align}\label{eq:second-first}
\mathrm{SDP}_{1,\mathrm{PPT}}(\mathcal{N},M)=\max &\quad d_{\bar A}d_B\cdot\mathrm{Tr}\left[J^\mathcal{N}_{\bar AB}Y_{\bar AB}\right] \\
s.t. &\quad \rho_{\bar A}\otimes\frac{1_B}{d_B}\succeq Y_{\bar AB}\succeq0,\;\mathrm{Tr}[\rho_{\bar A}]=1\\
& \quad M^2 \cdot Y_B=\frac{1_B}{d_B},\;\rho_{\bar A}\otimes\frac{1_B}{d_B}\succeq M\cdot Y_{\bar AB}^{T_B}\succeq-\rho_{\bar A}\otimes\frac{1_B}{d_B}.
\end{align}
The level $n=2$ reads as
\begin{align}\label{eq:second-level}
\mathrm{SDP}_{2,\mathrm{PPT}}(\mathcal{N},M)=\max &\quad d_{\bar A}d_B\cdot\mathrm{Tr}\left[\left(J^\mathcal{N}_{\bar AB_1}\otimes\Phi_{A\bar B_1}\right)W_{A\bar AB_1\bar B_1}\right] \\
s.t. &\quad W_{A\bar AB_1B_2\bar B_1\bar B_2}\succeq0,\;\mathrm{Tr}\left[W_{A\bar AB_1B_2\bar B_1\bar B_2}\right]=1\\
& \quad \mathcal{U}_{B_1B_2\bar B_1\bar B_2}^\pi\left(W_{A\bar AB_1B_2\bar B_1\bar B_2}\right)=W_{A\bar AB_1B_2\bar B_1\bar B_2}\;\forall\pi\in\Pi_2\\
& \quad W_{AB_1B_2\bar B_1\bar B_2}=\frac{1_A}{d_A}\otimes W_{B_1B_2\bar B_1\bar B_2},\;W_{A\bar AB_1B_2\bar B_1}=W_{A\bar AB_1\bar B_1}\otimes\frac{1_{B_2}}{d_B}\\
& \quad W_{A\bar AB_1B_2\bar B_1\bar B_2}^{T_{A\bar A}}\succeq0,\;W_{A\bar AB_1B_2\bar B_1\bar B_2}^{T_{B_2\bar B_2}}\succeq0.
\end{align}

Numerical evaluations of~\eqref{eq:second-first} and~\eqref{eq:second-level} can be found in Appendix~\ref{sec:numerics}.


\section{Conclusion}\label{sec:conclusion}

We have shown that quantum de Finetti theorems which impose linear constraints on the approximating state lead to converging SDP hierarchies for constrained bilinear optimization. As our main application, this gave efficiently computable outer bounds on the optimal quantum channel fidelity in approximate quantum error correction. In Appendix~\ref{sec:numerics}, we provide some numerical evidence that the resulting bounds are sometimes tight for low dimensional error models, but it would be desirable to do more extensive numerical studies for practically relevant settings. For example, it would be interesting to apply the techniques from \cite{rosset2018symdpoly} to automatically detect the symmetries in the problem in order to significantly improve the performance. One could also explore other operational settings in quantum information theory that are described in terms of jointly constrained semidefinite bilinear or multilinear programs (cf.~the related work~\cite{Huber18}).

On the mathematical side, it remains unclear if the linear constraint conditions in our quantum de Finetti theorem (Theorem~\ref{thm:main-deFinetti}) are minimal or could be further simplified. Recall that, for the linear constraint on system $B$, we had the condition
\begin{align}
\Gamma_{B_k\to C_B}(\rho_{B_1^k})&=\rho_{B_1^{k-1}}\otimes Y_{C_B}.
\end{align}
As in Example~\ref{example:weak-condition}, it is simple to see that only requiring $\Gamma_{B_k \to C_B}(\rho_{B_k}) = Y_{C_B}$ is not sufficient. However, the following weaker condition might be sufficient
\begin{align}
\Gamma_{B\to C_B}^{\otimes k}(\rho_{B_1^k})&=Y_{C_B}^{\otimes k}.
\end{align}
We leave this as an open question (also see the related works on de Finetti reductions \cite{fawzirenner14,brandao14}). Another mathematical question is to determine the optimal dimension dependence of the minimal distortion with side information (see Lemma~\ref{lem:ic-meas-side-info}). It would also be interesting to improve Corollary~\ref{cor:exchangeable-deFinetti} and give de Finetti theorems for quantum channels directly in terms of the diamond norm distance with a dimension dependence polynomial in $d_B$ and $k$. Finally, there are variants of quantum de Finetti theorems which provably lead to (exponentially) faster convergence for certain settings of the quantum separability problem~\cite{Christandl2007,brandao11,brandao13c}, and the consequences for approximate quantum error correction remain to be explored.


\section*{Acknowledgements}

We thank Fernando Brand\~ao, Matthias Christandl and Robert K\"onig for useful discussions, and Nengkun Yu for pointing out a mistake in a previous version of Corollary~\ref{cor:exchangeable-deFinetti} about channel de Finetti theorems in terms of the diamond norm distance.


\appendix

\section{Classically-assisted approximate quantum error correction}\label{sec:appendix}

\subsection{Setting}

It is often useful to add classical forward communication assistance to the problem of quantum error correction. The corresponding assisted channel fidelity is defined as follows.

\begin{definition}
Let $\mathcal{N}_{\bar A\to B}$ be a quantum channel and $M\in\mathbb{N}$. The LOCC(1)-assisted channel fidelity for message dimension $M$ is defined as\footnote{The term LOCC(1) stands for {\it local operations and one-way classical communication} from sender to receiver \cite{Chitambar14}.}
\begin{align}
F^{\mathrm{LOCC(1)}}(\mathcal{N},M):=\max &\quad F\Big(\Phi_{\bar BR},\sum_{i\in I}\big(\left(\mathcal{D}^i_{B\to\bar B}\circ\mathcal{N}_{\bar A\to B}\circ\mathcal{E}^i_{A\to\bar A}\right)\otimes\mathcal{I}_R\big)(\Phi_{AR})\Big)\\
s.t. &\quad \sum_{i\in I}\mathcal{E}^i_{A\to\bar A}\;\text{quantum channel with $\mathcal{E}^i_{A\to\bar A}$ cp for $i\in I$}\\
& \quad \mathcal{D}^i_{B\to\bar B}\;\text{quantum channel $\forall i\in I$},
\end{align}
where $\Phi_{AR}$ denotes the maximally entangled state on $AR$, cp is the abbreviation for completely positive, and we have $M=d_A=d_{\bar B}=d_R$.
\end{definition}

By the Choi-Jamio\l{}kowski isomorphism this can again be rewritten as a bilinear optimization.

\begin{lemma}
Let $\mathcal{N}_{\bar A\to B}$ be a quantum channel and $M\in\mathbb{N}$. Then, the LOCC(1)-assisted channel fidelity can be written as
\begin{align}
F^{\mathrm{LOCC(1)}}(\mathcal{N},M)=\max &\quad d_{\bar A}d_B\cdot\mathrm{Tr}\left[\left(J^\mathcal{N}_{\bar AB}\otimes\Phi_{A\bar B}\right)\left(\sum_{i\in I}E_{A\bar A}^i\otimes D_{B\bar B}^i\right)\right]\\
s.t. &\quad E_{A\bar A}^i\succeq0,\;D_{B\bar B}^i\succeq0\quad\forall i\in I\\
& \quad \sum_{i\in I}E_A^i=\frac{1_A}{d_A},\;D^i_B=\frac{1_B}{d_B}\quad\forall i\in I.
\end{align}
\end{lemma}

The proof follows similarly as in Lemma~\ref{lem:Choi} about plain quantum error correction, and is based on the manipulation of the objective function $F\Big(\Phi_{\bar BR},\sum_{i\in I}\big(\left(\mathcal{D}^i_{B\to\bar B}\circ\mathcal{N}_{\bar A\to B}\circ\mathcal{E}^i_{A\to\bar A}\right)\otimes\mathcal{I}_R\big)(\Phi_{AR})\Big)$ by using the Choi-Jamio\l{}kowski isomorphism. We have that $F^{\mathrm{LOCC(1)}}(\mathcal{N},M)$ is closely connected to the channel fidelity $F(\mathcal{N},M)$.

\begin{lemma}\label{lem:locc(1)-plain}
Let $\mathcal{N}$ be a quantum channel and $M\in\mathbb{N}$. Then, we have
\begin{align}
F^{\mathrm{LOCC(1)}}(\mathcal{N},M)\geq F(\mathcal{N},M)\geq \left(F^{\mathrm{LOCC(1)}}(\mathcal{N},M)\right)^2.
\end{align}
\end{lemma}

Asymptotically this corresponds to the well-known statement that forward classical communication assistance does not increase the capacity. 

\begin{proof}
The first inequality is trivial because the addition of a forward classical communication channel cannot decrease the channel fidelity. The fact that $\left(F^{\mathrm{LOCC(1)}}(\mathcal{N},M)\right)^2$ gives a lower bound on $F(\mathcal{N},M)$ can be seen from \cite[Proposition 4.5]{werner04}. Consider an arbitrary coding scheme for the quantum channel $\mathcal{N}$ assisted with a forward classical communication channel and call $\mathcal{F}_{\mathrm{LOCC(1)}}$ the channel fidelity obtained using that scheme. We then want to show that it is always possible to find a coding scheme for the quantum channel $\mathcal{N}$ alone allowing us to achieve a channel fidelity $\mathcal{F} \geq \mathcal{F}_{\mathrm{LOCC(1)}}^2$. Say we are able to send, through the forward classical communication channel, a symbol in the set $\{1,\dots,S\}$ with $S\in\mathbb{N}$. An arbitrary coding scheme for the assisted quantum channel can be modelled by a collection of instruments $\{\mathcal{E}^s_{A \rightarrow \bar A}\}_{s\in \{1,\dots,S\}}$, i.e., trace-nonincreasing cp maps summing up to a channel, and channels $\{\mathcal{D}^s_{B \rightarrow \bar B}\}_{s\in \{1,\dots,S\}}$. It is then easy to show that there must exist a symbol $\tilde s$ such that the fidelity of the map $\mathcal{D}^{\tilde s} \circ \mathcal{N} \circ \frac{\mathcal{E}^{\tilde s}}{e^{\tilde s}}$ is lower bounded by $\mathcal{F}_{\mathrm{LOCC(1)}}$, where the factor $e^{\tilde s}$ is chosen such that the completely positive map $\frac{\mathcal{E}^{\tilde s}}{e^{\tilde s}}$ becomes trace preserving with respect to the maximally mixed state $\frac{1_A}{d_A}$, as done in ~\cite[Proposition 5.1]{werner04}. Using the polar decomposition it is possible to find an isometric encoder $\mathcal{V}^{\tilde s} $ such that the channel fidelity $\mathcal{F}$ obtained using the coding scheme with encoder $\mathcal{V}^{\tilde s} $ and decoder $\mathcal{D}^{\tilde s} $ is lower bounded by the squared fidelity of the map $\mathcal{D}^{\tilde s} \circ \mathcal{N} \circ \frac{\mathcal{E}^{\tilde s}}{e^{\tilde s}}$. This implies $\mathcal{F} \geq \mathcal{F}_{\mathrm{LOCC(1)}} ^2$.
\end{proof}

We have the dimension bounds for the LOCC(1)-assisted setting. Notice that the following result readily implies Lemma \ref{lem:F-dimension}.

\begin{lemma}\label{lem:F-dimension-locc}
Let $\mathcal{N}_{\bar A\to B}$ be a quantum channel and $M\in\mathbb{N}$. Then, we have
\begin{align}
0\leq F^{\mathrm{LOCC(1)}}(\mathcal{N},M)\leq\min\left\{1,\left(\frac{d_{\bar A}}{M}\right)^2,\frac{d_B}{M}\right\}.
\end{align}
\end{lemma}

\begin{proof}
The lower bound is trivial. For the upper bounds, as in the proof of Lemma~\ref{lem:F-dimension}, we mainly use that for any sub-normalized bipartite quantum state $\rho_{XY}$ we have $d_X\cdot1_X\otimes\rho_Y\succeq\rho_{XY}$. Now, for the first upper bound note that $\frac{d_{\bar B}}{d_B}\cdot1_{B\bar B}=d_{\bar B}\cdot1_{\bar B}\otimes D_B^i\succeq D^i_{B\bar B}$ for all $i\in I$, and hence we get for the objective function (with $d_A=d_{\bar B}=M$)
\begin{align}
F^{\mathrm{LOCC(1)}}(\mathcal{N},M)&\leq d_{\bar A}d_{\bar B}\cdot\mathrm{Tr}\left[\left(J^\mathcal{N}_{\bar AB}\otimes\Phi_{A\bar B}\right)^{1/2}\left(\sum_{i\in I}E_{A\bar A}^i\otimes1_{B\bar B}\right)\left(J^\mathcal{N}_{\bar AB}\otimes\Phi_{A\bar B}\right)^{1/2}\right]\\
&=d_{\bar A}d_{\bar B}\cdot\mathrm{Tr}\left[\left(\frac{1_{\bar A}}{d_{\bar A}}\otimes\frac{1_A}{d_A}\right)\sum_{i\in I}E_{A\bar A}^i\right]=\mathrm{Tr}\left[\sum_{i\in I}E_{A\bar A}^i\right]=1.
\end{align}
For the second upper bound, note that from $E_{A\bar A}^i\succeq0,\;D_{B\bar B}^i\succeq0$ we get
\begin{align}
F^{\mathrm{LOCC(1)}}(\mathcal{N},M)\leq d_{\bar A}d_B\cdot\mathrm{Tr}\left[\left(J^\mathcal{N}_{\bar AB}\otimes\Phi_{A\bar B}\right)\left(\sum_{i\in I}E_{A\bar A}^i\otimes\sum_{j\in I}D_{B\bar B}^j\right)\right].
\end{align}
Now, we employ that $d_{\bar A}\cdot E^i_A\otimes1_{\bar A}\succeq E^i_{A\bar A}$ giving $\frac{d_{\bar A}}{d_A}\cdot1_{A\bar A}\succeq\sum_{i\in I}E^i_{A\bar A}$, which in turn leads to
\begin{align}
F^{\mathrm{LOCC(1)}}(\mathcal{N},M)&\leq\frac{d_{\bar A}^2d_B}{d_A}\cdot\mathrm{Tr}\left[\left(J^\mathcal{N}_{\bar AB}\otimes\Phi_{A\bar B}\right)\left(1_{A\bar A}\otimes\sum_{j\in I}D_{B\bar B}^j\right)\right]\\
&=\frac{d_{\bar A}^2d_B}{d_A}\cdot\mathrm{Tr}\left[\left(J^\mathcal{N}_B\otimes\frac{1_{\bar B}}{d_{\bar B}}\right)\sum_{j\in I}D_{B\bar B}^j\right]\\
&=\frac{d_{\bar A}^2d_B}{d_A^2d_{\bar B}}\cdot\mathrm{Tr}\left[J^\mathcal{N}_B\sum_{j\in I}D_B^j\right]=\frac{d_{\bar A}^2d_B}{d_A^2d_{\bar B}}\cdot\mathrm{Tr}\left[J^\mathcal{N}_Bd_A\frac{1_B}{d_B}\right]=\frac{d_{\bar A}^2}{d_Ad_{\bar B}}.
\end{align}
For the third upper bound, note that $1_{B\bar B}\succeq D^i_{B\bar B}$ and thus
\begin{align}
F^{\mathrm{LOCC(1)}}(\mathcal{N},M)&\leq d_{\bar A}d_B\cdot\mathrm{Tr}\left[\left(J^\mathcal{N}_{\bar AB}\otimes\Phi_{A\bar B}\right)\left(\sum_{i\in I}E_{A\bar A}^i\otimes1_{B\bar B}\right)\right]\\
&=d_{\bar A}d_B\cdot\mathrm{Tr}\left[\left(\frac{1_{\bar A}}{d_{\bar A}}\otimes\frac{1_A}{d_A}\right)\sum_{i\in I}E_{A\bar A}^i\right]=\frac{d_B}{d_A}\cdot\mathrm{Tr}\left[\sum_{i\in I}E_{A\bar A}^i\right]=\frac{d_B}{d_A}.
\end{align}
\end{proof}


\subsection{Hierarchy of outer bounds}

By removing one of the two conditions in Theorem~\ref{lem:sep-deFinetti}, we get the following approximation for the set of LOCC(1) channels\,---\,stated in terms of the corresponding Choi states.

\begin{proposition}\label{lem:locc(1)-deFinetti}
Let $\rho_{A\bar A(B\bar B)_1^n}$ be a quantum state with
\begin{align}
\rho_{A\bar A(B\bar B)_1^n}=\mathcal{U}_{(B\bar B)_1^n}^\pi(\rho_{A\bar A(B\bar B)_1^n})\;\forall\pi\in\mathfrak{S}_n,\quad\rho_A=\frac{1_A}{d_A},\quad\rho_{(B\bar B)_1^{n-1}B_n}=\rho_{(B\bar B)_1^{n-1}}\otimes\frac{1_{B_n}}{d_B}.
\end{align}
Then, we have for $0<k<n$ that $\left\|\rho_{A\bar A(B\bar B)_1^{k}}-\sum_{i\in I}\sigma_{A\bar A}^i\otimes\left(\omega_{B\bar B}^i\right)^{\otimes k}\right\|_1$ is upper bounded by the same term as in Theorem~\ref{lem:sep-deFinetti}, where $\omega_{B\bar B}^i\succeq0$ with $\omega_B^i=\frac{1_B}{d_B}$ and $\sigma_{A\bar A}^i\succeq0$ with $\sum_{i\in I}\sigma_A^i=\frac{1_A}{d_A}$.
\end{proposition}

The $n$-th level of the SDP hierarchy then becomes
\begin{align}
\mathrm{SDP}^{\mathrm{LOCC(1)}}_n(\mathcal{N},M):=\max &\quad d_{\bar A}d_B\cdot\mathrm{Tr}\left[\left(J^\mathcal{N}_{\bar AB_1}\otimes\Phi_{A\bar B_1}\right)Z_{A\bar AB_1\bar B_1}\right]\\
s.t. &\quad Z_{A\bar A(B\bar B)_1^n}\succeq0,\;\mathcal{U}_{(B\bar B)_1^n}^\pi\left(Z_{A\bar A(B\bar B)_1^n}\right)=Z_{A\bar A(B\bar B)_1^n}\quad\forall\pi\in\mathfrak{S}_n\\
& \quad Z_{AB_1^n}=\frac{1_{AB_1^n}}{d_Ad_B^n},\;Z_{A\bar A\left(B\bar B\right)_1^{n-1}B_n}= Z_{A\bar A\left(B\bar B\right)_1^{n-1}}\otimes \frac{1_{B_n}}{d_B}.
\end{align}
By inspection, the only difference between $\mathrm{SDP}_n(\mathcal{N},M)$ and $\mathrm{SDP}^{\mathrm{LOCC(1)}}_n(\mathcal{N},M)$ is the weakened second to last condition. The asymptotic convergence follows immediately from Proposition~\ref{lem:locc(1)-deFinetti}.

\begin{theorem}\label{thm:convergence-locc(1)}
Let $\mathcal{N}$ be a quantum channel and $n,M\in\mathbb{N}$. Then, we have
\begin{align}
\mathrm{SDP}^{\mathrm{LOCC(1)}}_{n+1}(\mathcal{N},M)\leq\mathrm{SDP}^{\mathrm{LOCC(1)}}_n(\mathcal{N},M)\quad\text{and}\quad F^{\mathrm{LOCC(1)}}(\mathcal{N},M)=\lim_{n\to\infty}\mathrm{SDP}^{\mathrm{LOCC(1)}}_n(\mathcal{N},M).
\end{align}
\end{theorem}

Note that for $\mathrm{SDP}^{\mathrm{LOCC(1)}}_n(\mathcal{N},M)$ we slightly strengthened the last two conditions by including some more $A$- and $B$-systems in the conditions compared to the minimal conditions
\begin{align}
Z_A=\frac{1_A}{d_A}\quad\text{and}\quad Z_{\left(B\bar B\right)_1^{n-1}B_n}=\frac{1_{B_n}}{d_B}\otimes Z_{\left(B\bar B\right)_1^{n-1}}
\end{align}
needed for Proposition~\ref{lem:locc(1)-deFinetti}. By an iterative argument the last condition implies in particular that
\begin{align}
Z_{A\bar AB_1^n\bar B_1}=\frac{1_{B_2^n}}{d_B^n}\otimes Z_{A\bar AB_1\bar B_1},
\end{align}
which together with the other three conditions in $\mathrm{SDP}^{\mathrm{LOCC(1)}}_n(\mathcal{N},M)$ then corresponds to the notion of {\it extendible channels} from~\cite[Definition 5]{kaur18} (also see~\cite{duan16} for similar conditions). We note, however, that when relaxing the conditions to $n$-extendible channels our proofs for the asymptotic convergence of the resulting outer bounds do not apply.

The SDP relaxations again behave naturally in the sense that they are upper bounded by one.

\begin{lemma}\label{lem:ineqspnlocc}
Let $\mathcal{N}$ be a quantum channel and $n,M\in\mathbb{N}$. Then, we have
\begin{align}
0\leq\mathrm{SDP}^{\mathrm{LOCC(1)}}_n(\mathcal{N},M)\leq1.
\end{align}
\end{lemma}

\begin{proof}
The lower bound is trivial. For the upper bound, by the monotonicity in $n$ (Theorem~\ref{thm:convergence-locc(1)}) it is enough to restrict to $n=1$. As in the proof of Lemma~\ref{lem:F-dimension-locc}, we make use of $\frac{d_{\bar B}}{d_B}\cdot Z_{A\bar A}\otimes1_{B_1\bar B_1}\succeq Z_{A\bar AB_1\bar B_1}$. This again gives
\begin{align}
\mathrm{SDP}^{\mathrm{LOCC(1)}}_1(\mathcal{N},M)\leq d_{\bar A}d_B\cdot\mathrm{Tr}\left[\left(J^\mathcal{N}_{\bar AB}\otimes\Phi_{A\bar B}\right)\frac{d_{\bar B}}{d_B}\cdot Z_{A\bar A}\otimes1_{B_1\bar B_1}\right]=1.
\end{align}
\end{proof}

We can again add PPT constraints and we denote the resulting relaxations by $\mathrm{SDP}^{\mathrm{LOCC(1)}}_{n,\mathrm{PPT}}(\mathcal{N},M)$. In the following we study more closely these levels $\mathrm{SDP}_{n,\mathrm{PPT}}^{\mathrm{LOCC(1)}}(\mathcal{N},M)$, which are our tightest outer bound relaxations on the LOCC(1)-assisted channel fidelity. We find
\begin{align}
\mathrm{SDP}^{\mathrm{LOCC(1)}}_{1,\mathrm{PPT}}(\mathcal{N},M)=\max &\quad d_{\bar A}d_B\cdot\mathrm{Tr}\left[\left(J^\mathcal{N}_{\bar AB}\otimes\Phi_{A\bar B}\right)Z_{A\bar AB\bar B}\right] \\
s.t. &\quad Z_{A\bar AB\bar B}\succeq0,\;Z_{A\bar AB\bar B}^{T_{B\bar B}}\succeq0\\
& \quad Z_{AB}=\frac{1_{AB}}{d_Ad_B},\;Z_{A\bar AB}=Z_{A\bar A}\otimes\frac{1_B}{d_B}.
\end{align}
This is exactly the SDP outer bound found in~\cite[Section IV]{matthews14}, which simplifies to
\begin{align}
\mathrm{SDP}^{\mathrm{LOCC(1)}}_{1,\mathrm{PPT}}(\mathcal{N},M)=\max &\quad d_{\bar A}d_B\cdot\mathrm{Tr}\left[J^\mathcal{N}_{\bar AB}X_{\bar AB}\right] \\
s.t. &\quad \rho_{\bar A}\otimes\frac{1_B}{d_B}\succeq X_{\bar AB}\succeq0,\;\mathrm{Tr}[\rho_{\bar A}]=1\\
& \quad \rho_{\bar A}\otimes\frac{1_B}{d_B}\succeq M\cdot X_{\bar AB}^{T_B}\succeq-\rho_{\bar A}\otimes\frac{1_B}{d_B}.
\end{align}
By inspection, this corresponds to $\mathrm{SDP}_{1,\mathrm{PPT}}(\mathcal{N},M)$ but with one missing constraint, namely $M^2X_B=\frac{1_B}{d_B}$. For $n=2$ we get
\begin{align}
\mathrm{SDP}^{\mathrm{LOCC(1)}}_{2,\mathrm{PPT}}(\mathcal{N},M)=\max &\quad d_{\bar A}d_B\cdot\mathrm{Tr}\left[\left(J^\mathcal{N}_{\bar AB_1}\otimes\Phi_{A\bar B_1}\right)Z_{A\bar AB_1\bar B_1}\right] \\
s.t. &\quad Z_{A\bar AB_1B_2\bar B_1\bar B_2}\succeq0,\;Z_{A\bar AB_1B_2\bar B_1\bar B_2}^{T_{A\bar A}}\succeq0,\;Z_{A\bar AB_1B_2\bar B_1\bar B_2}^{T_{B_2\bar B_2}}\succeq0\\
& \quad \mathcal{U}_{B_1B_2\bar B_1\bar B_2}^\pi\left(Z_{A\bar AB_1B_2\bar B_1\bar B_2}\right)=Z_{A\bar AB_1B_2\bar B_1\bar B_2}\quad\forall\pi\in\Pi_2\\
& \quad Z_{AB_1B_2}=\frac{1_{AB_1B_2}}{d_Ad_B^2},\;Z_{A\bar AB_1B_2\bar B_1}=Z_{A\bar AB_1\bar B_1}\otimes\frac{1_{B_2}}{d_B},
\end{align}
and we recover the exact same conditions as for the notion of extendible channels~\cite[Definition 5]{kaur18}.


\section{Numerical examples}\label{sec:numerics}

\subsection{Methods}

In the following we present the proof of concept numerics we implemented to test the low levels of our hierarchy for the application of approximate quantum error correction. The experiments have been done in MATLAB using the QETLAB library~\cite{referenceQETLAB}, CVX~\cite{referenceCVX}, MOSEK~\cite{mosek17}, and SDPT3~\cite{Toh2012}.\footnote{All our code is available at https://github.com/FrancescoBorderi/Quantum-SDPs.} As discussed in Lemma~\ref{lem:rank-loop}, the authors of~\cite{Navascues09} gave a rank loop condition to certify that a certain level of the hierarchy already gives the optimal value. We restate the condition here in the exact form needed for approximate quantum error correction.

\begin{lemma}\label{lem:rank-loop-channel}
Let $W_{A\bar A(B\bar B)_1^n}=\mathcal{U}_{(B\bar B)_1^n}^\pi\left(W_{A\bar A(B\bar B)_1^n}\right)$ for all $\pi\in\mathfrak{S}_n$ and fixed $0\leq k\leq n$ such that $W_{A\bar A(B\bar B)_1^n}^{T_{(B\bar B)_{k+1}^n}}\succeq0$. If we have
\begin{align}
\mathrm{rank}\left(W_{A\bar A(B\bar B)_1^n}\right) \leq \max\left\{\mathrm{rank}\left(W_{A\bar A (B\bar B)_1^k}\right),\,\mathrm{rank}\left(W_{ (B\bar B)_{k+1}^n}\right)\right\},
\end{align}
then $W_{A\bar AB\bar B}$ is separable with respect to the partition $A \bar A | B \bar B$.
\end{lemma}

Using Lemma~\ref{lem:rank-loop-channel} it is in principle possible to, e.g., certify the optimality of the first level using the second level of our hierarchy. Moreover, if the criterion is fulfilled it will also allow us to extract the actual encoder and decoder of the optimal quantum error correction code. However, in order to facilitate the search for solutions having rank loops we need to look for low rank solutions $W_{A\bar A(B\bar B)_1^n}$. It is not possible to directly write a rank condition into our semidefinite programs because rank constraints are not convex. In addition, SDP solvers typically give high rank solutions since they tend to look for solutions at the interior of the convex set.\footnote{We noticed that SDPT3 compared to MOSEK gives results having in general lower rank.} Nevertheless, a possible strategy is to find a solution $W_{A\bar A(B\bar B)_1^n}$ and then employ a heuristic to minimize the rank while keeping the hierarchy constraints. The heuristic we found the most effective for our purposes was the log-det method described in~\cite{Fazel03}. The idea is to minimize the first-order Taylor series expansion of
\begin{align}
\log \det \left(W_{A\bar A(B\bar B)_1^n} +\delta\cdot1\right),
\end{align}
which is used as a smooth surrogate for $\mathrm{rank}\big(W_{A\bar A(B\bar B)_1^n}\big)$ and $\delta>0$ is a small regularization constant. The procedure is iterative, meaning that we start from $W_0 = 1_{A\bar A(B\bar B)_1^n}$, then compute $W_1$ minimizing the log-det objective function, and so on. In particular, the choice $W_0 = 1_{A\bar A(B\bar B)_1^n}$ connects the method to the trace heuristic, which is known to be an effective heuristic for rank reduction. We stop after a certain number $l$ of iterations and then we find a solution $W_l$ having hopefully lower rank than the original $\mathrm{rank}\big(W_{A\bar A(B\bar B)_1^n}\big)$.


\subsection{Qubit Channels}

We computed SDP relaxations in the plain coding setting for all the most common qubit channels: depolarizing, amplitude damping, bit flip, phase flip, bit-phase flip, Werner-Holevo and generalized Werner-Holevo channel. We found the upper bounds
\begin{align}
\mathrm{SDP}_{1,\mathrm{PPT}}(\mathcal{N}_2,2) = \mathrm{SDP}_{2,\mathrm{PPT}}(\mathcal{N}_2,2) = \mathrm{SDP}_{3,\mathrm{PPT}}(\mathcal{N}_2,2)=\mathrm{SDP}_1(\mathcal{N}_2,2) = \mathrm{SDP}_2(\mathcal{N}_2,2) = \mathrm{SDP}_3(\mathcal{N}_2,2),
\end{align}
where the subscript in $\mathcal{N}_2$ refers to the two-dimensional input and output of the channel. These identities also remain true for random qubit channels and one might then conjecture that for qubit channels indeed already $\mathrm{SDP}_1(\mathcal{N}_2,2)$ captures $F(\mathcal{N},2)$.

For the qubit depolarizing channel the trivial coding scheme is known to be optimal and we retrieve this result using the rank loop condition of the second level based on the log-det method. Similarly, for the qubit bit flip channel with parameter $p=0.1$ we find a rank-one state solution of the second level using again the log-det method, implying that the rank loop condition holds. In this case the solution is not just the state associated with the trivial coding scheme via the Choi isomorphism but the resulting encoder/decoder pair with optimal fidelity $0.9$ is given by the unitary channels with Kraus matrices $U_E= - \ketbra{1}{0} + \ketbra{0}{1}$ and $U_D=  \ketbra{0}{0} - \ketbra{1}{1}$, respectively. Note that the trivial coding scheme is largely suboptimal for a qubit bit flip channel with $p=0.1$, as the corresponding fidelity is $0.1$.

\begin{figure} [ht]
	\includegraphics[width=0.7\linewidth]{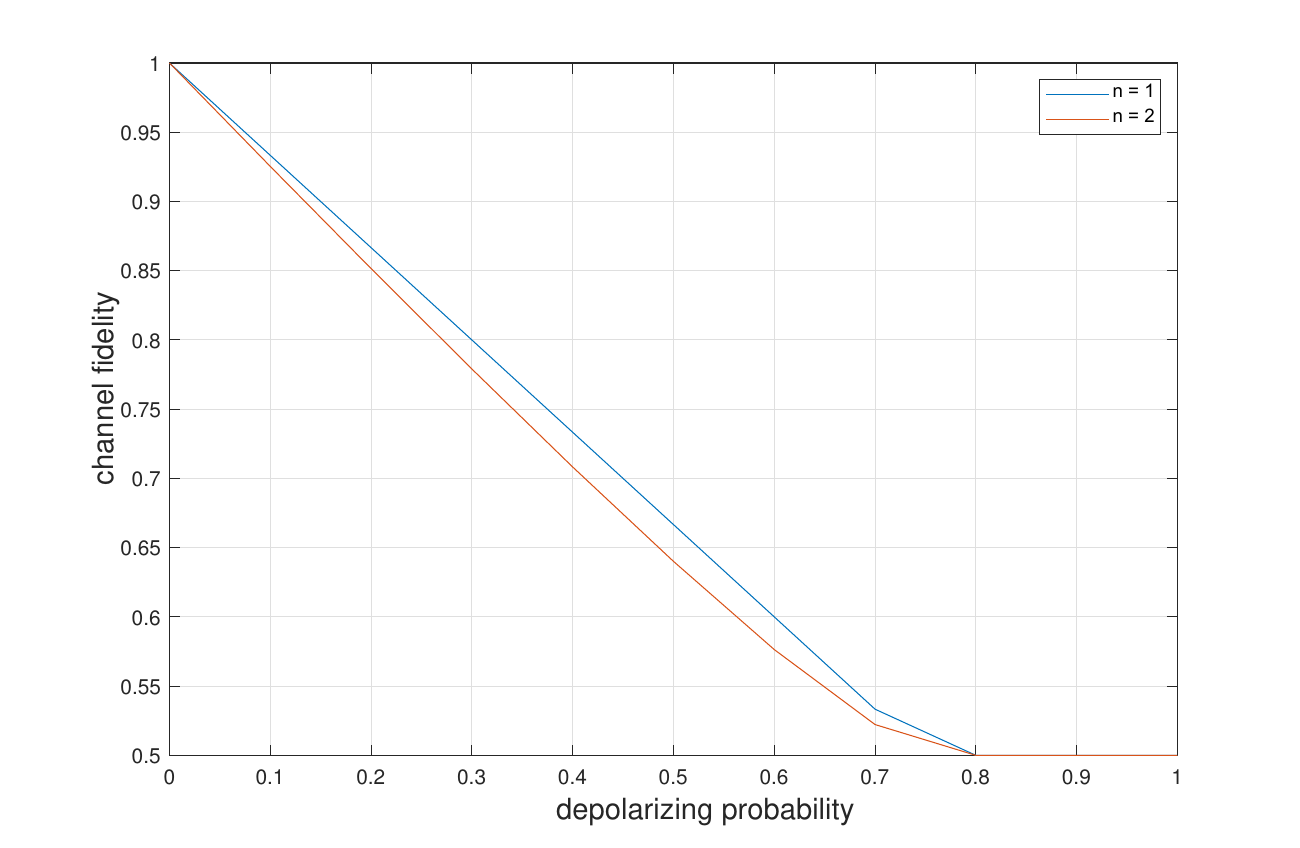}
	\caption{Comparison of the SDP upper bounds $n=1,2$ on the channel fidelity of the 3-dimensional depolarizing channel for LOCC(1)-assisted coding (see Appendix~\ref{sec:appendix}). We see an improvement for the second level for $p \in (0,0.8)$.}
	\label{fig:depL}
\end{figure}

\begin{figure} [ht]
	\includegraphics[width=0.7\linewidth]{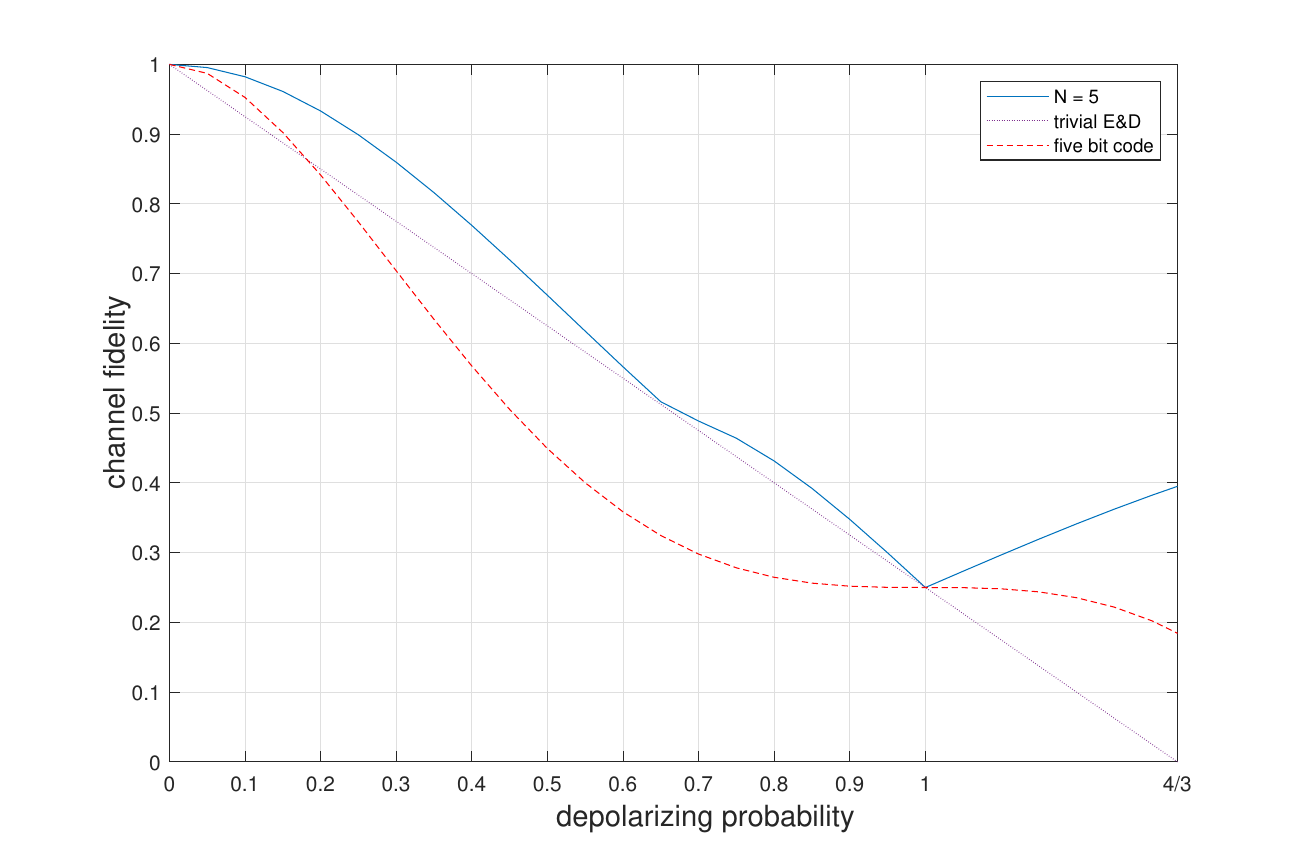}
	\caption{Comparison of the SDP upper bound $n=1$ on the channel fidelity for five repetitions of the qubit depolarizing channel in the plain coding setting, with the trivial coding scheme and the 5 qubit stabilizer code from~\cite{Bennett96}. Notice the intersection of the 5 qubit code and the trivial scheme in the region $p\in(0.1,0.2)$ and the singular behaviour of the first level in the region $p\in(0.6,0.7)$. In addition, for $p \in [1,4/3]$ the behaviour of the first level seems to match exactly with the lower bound obtained with an iterative seesaw algorithm reported in Figure 3.7 of~\cite[Chapter 3]{Werner05}.}
	\label{fig:depfive}
\end{figure}

\begin{figure} [ht]
	\includegraphics[width=0.7\linewidth]{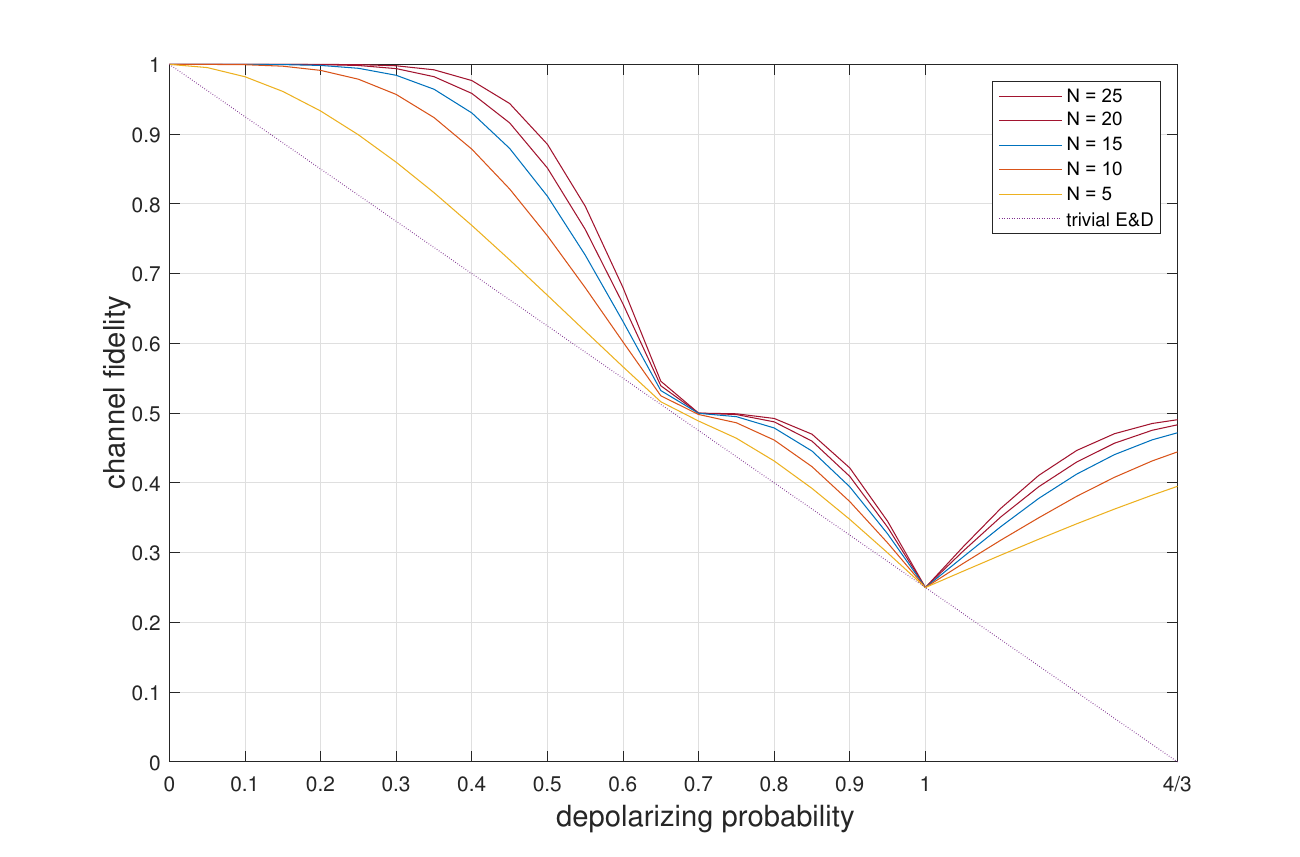}
	\caption{Comparison of the SDP upper bound $n=1$ on the channel fidelity for 5, 10, 15, 20, 25 repetitions of the 2-dimensional depolarizing channel in the plain coding setting. Notice that the singular behaviour of the first level in the region $p \in (0.6,0.7)$ is even more accentuated with the increase of the number of repetitions.}
	\label{fig:depseveral}
\end{figure}

\begin{figure} [ht]
	\includegraphics[width=0.7\linewidth]{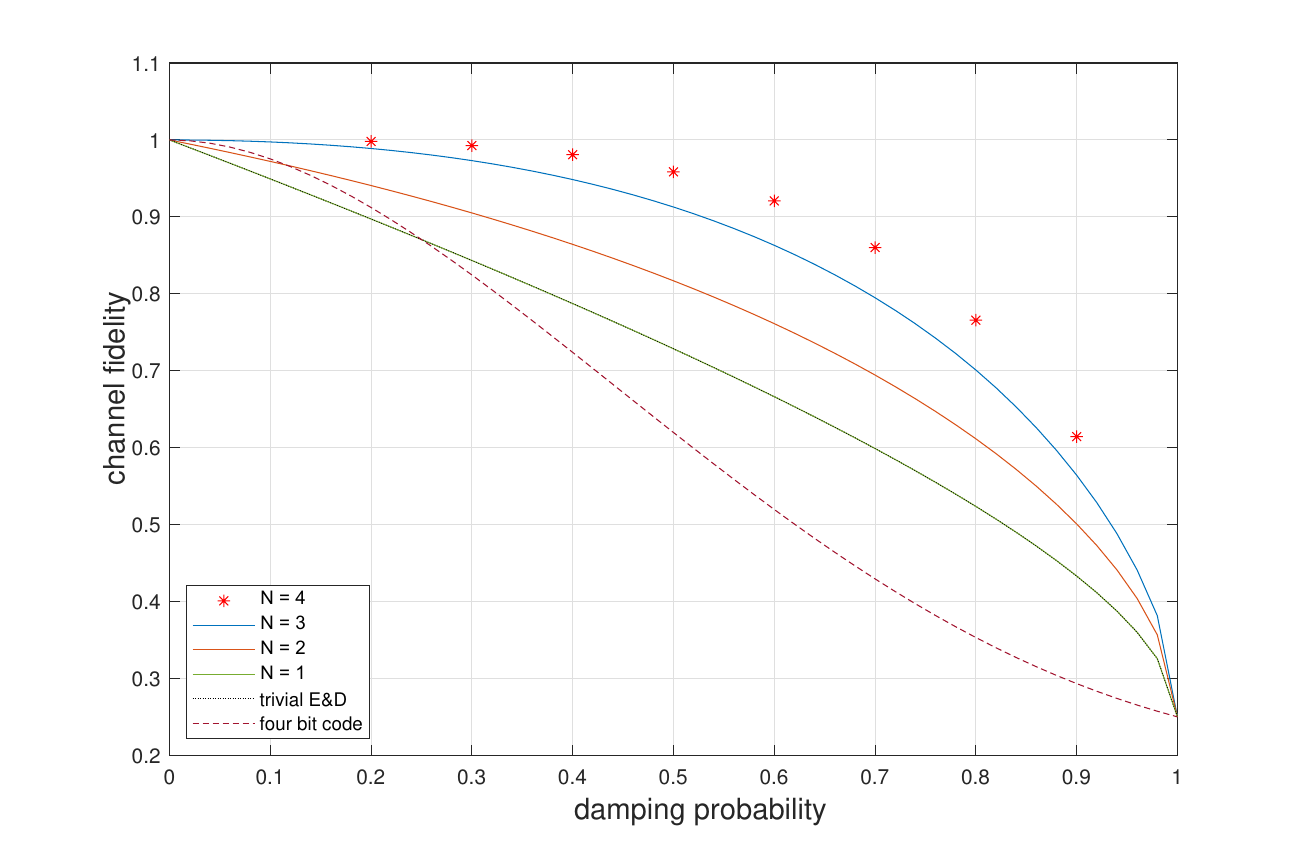}
	\caption{Comparison of the SDP upper bound $n=1$ on the channel fidelity of the qubit amplitude damping channel for 1,2,3 and 4 repetitions in the plain coding setting, as well as the trivial encoder and decoder and the 4 qubit code from~\cite{Leung97}.}
	\label{fig:ampdamp}
\end{figure}


\subsection{Qutrit Channels}\label{qutrit}

We computed SDP relaxations in the plain coding setting for the following qutrit channels: depolarizing, Werner-Holevo and generalized Werner-Holevo channel. We found the upper bounds $\mathrm{SDP}_{1,\mathrm{PPT}}(\mathcal{N}_3,2) = \mathrm{SDP}_{2,\mathrm{PPT}}(\mathcal{N}_3,2)$ and this identity also remains true for random qutrit channels. Removing the PPT conditions, however, we found qutrit channels $\mathcal{N}_3$ such that $\mathrm{SDP}_2(\mathcal{N}_3,2) < \mathrm{SDP}_1(\mathcal{N}_3,2)$.


\subsection{Depolarizing channel}

The depolarizing channel for $p \in [0,4/3]$ is given as
\begin{align}
Dep_d:\rho_{\bar A} \mapsto p\cdot\tr[\rho_{\bar A}]\frac{1_B}{d_B} + (1-p)\cdot\rho_{B},
\end{align}
where $d$ denotes the dimension of the input and output. Notice that even though often the channel is only studied for $p \in [0,1]$ where we can interpret $p$ as a depolarizing probability, the above expression also represents a channel for $p \in (1,4/3]$ (as, e.g., discussed in~\cite[Chapter 3]{Werner05}). We find that
\begin{align}
\mathrm{SDP}_{1,\mathrm{PPT}}(Dep_2,2) = \mathrm{SDP}_{2,\mathrm{PPT}}(Dep_2,2)=\mathrm{SDP}_{1,\mathrm{PPT}}(Dep_3,2) = \mathrm{SDP}_{2,\mathrm{PPT}}(Dep_3,2).
\end{align}
However, in Section \ref{qutrit} we found that in general removing the PPT conditions allows us to see a difference for the first two levels. This behaviour is not shown by the qutrit depolarizing channel, probably due to its highly symmetrical structure. We computed the upper bound for LOCC(1) coding (see Appendix~\ref{sec:appendix}) and found for $p \in (0,0.8)$ that
\begin{align}
\mathrm{SDP}^{\mathrm{LOCC(1)}}_{2,\mathrm{PPT}}(Dep_2,2)=\mathrm{SDP}^{\mathrm{LOCC(1)}}_{1,\mathrm{PPT}}(Dep_2,2),\;\text{while}\;\mathrm{SDP}^{\mathrm{LOCC(1)}}_{2,\mathrm{PPT}}(Dep_3,2)<\mathrm{SDP}^{\mathrm{LOCC(1)}}_{1,\mathrm{PPT}}(Dep_3,2).
\end{align}
We compared, for the plain coding setting, the $n=1$ level for five repetitions of the qubit depolarizing channel with the fidelity of the trivial coding scheme, as well as the 5 qubit stabilizer code from~\cite{Bennett96}. In particular, following~\cite{wang17} we exploited the symmetries of the qubit depolarizing channel to get the linear program
\begin{align}
\mathrm{SDP}_{1,\mathrm{PPT}}\left(Dep_2^{\otimes N},2\right)=\max &\quad \sum_{i=0}^N \binom{N}{i}\left(1-\frac{3p}{4}\right)^i {\left(\frac{3p}{4}\right)}^{N-i} m_i \\
s.t. &\quad 0 \le m_i \le 1 \quad i\in \{0,\dots,N\}\\
& \quad -\frac{1}{2} \leq \sum_{i = 0}^N x_{i,k} m_i \le \frac{1}{2} \quad k\in \{0,\dots,N\}\\
& \quad \sum_{i=0}^{N}\binom{N}{i} 3^{N-i} m_i = 2^{2N-2}.
\end{align}
where $x_{i,k}=\frac{1}{d^N}\sum_{r= \max\{0,i+k-N\}}^{\min\{i,k\}} \binom{k}{r} \binom{N-k}{i-r}  (-1)^{i-r} (d-1)^{k-r} (d+1)^{N-k+r-i}$ with $i,k\in \{0,\dots,N\}$. Notice that the number of variables is an affine function of $N$. The results are reported in Figure \ref{fig:depfive}. Comparing these with Figure 3.7 in \cite[Chapter 3]{Werner05}, it seems that the first level of the hierarchy matches their lower bounds in the region $p \in [1,4/3]$. Notice the intersection of the five qubit code and the trivial coding scheme in the region $p\in(0.1,0.2)$ and the singular behaviour in the region $p\in(0.6,0.7)$. We have also examined five, ten, fifteen, twenty and twenty five repetitions of the qubit depolarizing channel, again using the above linear program. The results are shown in Figure \ref{fig:depseveral}. Notice that the singular behaviour noted in Figure~\ref{fig:depfive} is now even more accentuated when increasing the number of repetitions.


\subsection{Amplitude damping channel}

The qubit amplitude damping channel with damping probability $\gamma \in [0,1]$ is given as
\begin{align}
Amp_{\gamma}:\rho_{\bar A} \rightarrow E^0_B \rho_{B}{E^0_B}^\dagger + E^1_B \rho_{B}{E^1_B}^\dagger,\;\text{where}\;E^0_B = \ketbra{0}{0} + \sqrt{1-\gamma}\ketbra{1}{1},\;E^1_B = \sqrt{\gamma}\ketbra{0}{1}.
\end{align}
We compared the results given by one, two, three, and four repetitions of the channel for the level $n=1$. The bounds are shown in Figure \ref{fig:ampdamp}, compared with the fidelity of the trivial coding scheme, and the 4 qubit code from~\cite{Leung97}. Notice the overlap between the first level of the hierarchy and the trivial coding scheme for the one-shot setting. Comparing these results with Figure 3.12 in~\cite[Chapter 3]{Werner05} we see that there is gap between their lower bounds (that significantly improve on the trivial coding scheme) and our upper bounds.


\section{Worst case error criteria}\label{sec:diamondNorm}

\subsection{Setting}

So far we have used the channel fidelity from Definition \ref{def:plain-coding} as the measure to study approximate quantum error correction\,---\,which corresponds to the average error case. In this appendix, we consider the diamond norm to study the worst case error and we find a program for which the hierarchy can be used to generate, in this case, lower bounds. We prove the sequence of semidefinite relaxations do in fact converge to the exact value of the original optimization program.

\begin{definition}\label{def:plain-worstcoding}
Let $\mathcal{N}_{\bar A\to B}$ be a quantum channel and $M \in \mathbb{N}$, with $M=d_{A}=d_{\bar B}$. The channel distance is defined as
\begin{align}
\Delta(\cN,M):=\min &\;\frac{1}{2}\left\|\mathcal{D}_{B\to\bar B}\circ\mathcal{N}_{\bar A\to B}\circ\mathcal{E}_{A\to\bar A}-\mathcal{I}_{A\to\bar B}\right\|_{\Diamond} \\
\mathrm{s.t.} &\;\mathcal{D}_{B\to\bar B},\mathcal{E}_{A\to\bar A}\;\text{quantum channels.}
\end{align}
\end{definition}

The following lemma writes the channel distance as given in Definition~\ref{def:plain-worstcoding} in terms of the Choi matrices of the encoder $\mathcal{E}_{A\to\bar A}$ and decoder $\mathcal{D}_{B\to\bar B}$, respectively.

\begin{lemma}\label{lem:Choi-worst}
Let $\mathcal{N}_{\bar A\to B}$ be a quantum channel and $M\in\mathbb{N}$. Then, we have that
\begin{align}
&\Delta(\cN,M)=\min\;\lambda \\
\mathrm{s.t.} &\; \;E_{A\bar A}\succeq0,\,E_A=\frac{1_A}{d_A},\;D_{B\bar B}\succeq0,\,D_B=\frac{1_B}{d_B}\\
&\;Z_{A\bar B}\succeq0,\;\frac{\lambda}{d_A}\cdot1_A\succeq Z_A\\
&\; Z_{A\bar B}+\Phi_{A\bar B}\succeq d_{\bar A}d_{B}\cdot\mathrm{Tr}_{\bar A B}\left[\left( 1_{A} \otimes J^{\mathcal{N}}_{\bar A B}\otimes 1_{\bar B})(E_{A\bar A}\otimes D_{B\bar B}\right)\right],
\end{align}
where $J^{\mathcal{N}}_{\bar A B}$ denotes the Choi matrix of $\mathcal{N}_{\bar A\to B}$.
\end{lemma}

\begin{proof}
Following \cite{v005a011}, the channel distance $\Delta(\cN,M)$ can be written as
\begin{align}
\Delta(\cN,M)=\min&\;\left\|Z_A \right\|_{\infty}\\
\mathrm{s.t.} &\;\mathcal{D}_{B\to\bar B},\mathcal{E}_{A\to\bar A}\;\text{quantum channels} \\
&\;Z_{A\bar B}\succeq 0,\;Z_{A\bar B}\succeq d_A \cdot J^{\mathcal{D}\circ\mathcal{N}\circ\mathcal{E}-\mathcal{I}}_{ A \bar B}.
\end{align}
We simplify
\begin{align}
J^{\mathcal{D}\circ\mathcal{N}\circ\mathcal{E}-\mathcal{I}}_{ A \bar B} &= J^{\mathcal{D}\circ\mathcal{N}\circ\mathcal{E}}_{ A \bar B} - J^{\mathcal{I}}_{ A \bar B} 
= J^{\mathcal{D}\circ\mathcal{N}\circ\mathcal{E}}_{ A \bar B} - \Phi_{ A \bar B},
\end{align}
write for the infinity norm $\left\|Z_A \right\|_{\infty}=\min\left\{\lambda\in \mathbb{R}: \lambda \cdot 1_A \succeq Z_A\right\}$, and relabel $\frac{Z_{A\bar B}}{d_A}$ as $Z_{A\bar B}$, leading to
\begin{align}
\Delta(\cN,M)=\min&\;\lambda\label{eq:watrous}\\
\mathrm{s.t.} &\;\mathcal{D}_{B\to\bar B},\mathcal{E}_{A\to\bar A}\;\text{quantum channels}\\
&\;Z_{A\bar B}\succeq 0,  \frac{\lambda}{d_A} \cdot 1_A \succeq Z_A \\
&\; Z_{A\bar B} + \Phi_{ A \bar B} \succeq  J^{\mathcal{D}\circ\mathcal{N}\circ\mathcal{E}}_{A\bar B}.
\end{align}
Following \cite{matthews14} and in particular \cite[Equation 7]{wang17-2}, we have the Choi state
\begin{align}
J^{\mathcal{D}\circ\mathcal{N}\circ\mathcal{E}}_{ A \bar B}=d_{\bar A}d_{B}\cdot\mathrm{Tr}_{\bar A B}\left[ \left( 1_{A} \otimes J^{\mathcal{N}}_{\bar A B}\otimes 1_{\bar B}\right)\left(J^{\mathcal{E}}_{ A \bar A}\otimes J^{\mathcal{D}}_{ B \bar B}\right)\right]
\end{align}
and writing $J^{\mathcal{E}}_{ A \bar A} = E_{A\bar A}$ as well as $J^{\mathcal{D}}_{ B \bar B} = D_{B\bar B} $ concludes the proof.
\end{proof}


\subsection{Hierarchy of lower bounds}

Similarly as in Section \ref{subsectionHierarchy}, we define a hierarchy of semidefinite programs labelled by an index $n$. Our framework directly applies as the structure of the optimization problem derived in Lemma \ref{lem:Choi-worst} involves the tensor product $E_{A\bar A}\otimes D_{B\bar B}$. The $n$-th level of the SDP hierarchy then generates the lower bounds $\mathrm{SDP}^\Delta_n(\mathcal{N},M)$ for the distance $\Delta(\cN,M)$ as
\begin{align}
\mathrm{SDP}^\Delta_n(\mathcal{N},M):=\min&\;\lambda \\
\mathrm{s.t.} &\; W_{A\bar A(B\bar B)_1^n}\succeq0,\;\mathrm{Tr}\left[W_{A\bar A(B\bar B)_1^n}\right]=1\\
&\;W_{A\bar A(B\bar B)_1^n}=\mathcal{U}_{(B\bar B)_1^n}^\pi\left(W_{A\bar A(B\bar B)_1^n}\right)\;\forall\pi\in\mathfrak{S}_n\\
&\;W_{A(B\bar B)_1^n}=\frac{1_A}{d_A}\otimes W_{(B\bar B)_1^n},\;W_{A\bar A(B\bar B)_1^{n-1}B_n}=W_{A\bar A(B\bar B)_1^{n-1}}\otimes\frac{1_{B_n}}{d_B}\\
&\; Z_{A\bar B}\succeq0,\;\frac{\lambda}{d_A}\cdot1_A\succeq Z_A\\
&\; Z_{A\bar B}+\Phi_{A\bar B}\succeq d_{\bar A}d_{B}\cdot\mathrm{Tr}_{\bar A B}\left[\left(1_{A} \otimes J^{\mathcal{N}}_{\bar A B}\otimes 1_{\bar B}\right)W_{A\bar AB\bar B}\right].
\end{align}
We can also add PPT constraints and denote the resulting relaxations by $\mathrm{SDP}^\Delta_{n,\mathrm{PPT}}(\mathcal{N},M)$. The following theorem states the convergence of the hierarchy.

\begin{theorem}\label{thm:convergence-diamond}
Let $\mathcal{N}$ be a quantum channel and $n,M\in\mathbb{N}$. Then, we have
\begin{align}
0\leq\Delta(\cN,M)-\mathrm{SDP}^\Delta_{n}(\mathcal{N},M)\leq\frac{\mathrm{poly}(d)}{\sqrt{n}}\quad\text{implying}\quad\Delta(\cN,M)=\lim_{n\to\infty}\mathrm{SDP}^\Delta_{n}(\mathcal{N},M),
\end{align}
where $d=\max\{d_A,d_{\bar A},d_B,d_{\bar B}\}$.
\end{theorem}

\begin{proof}
The bound $0\leq\Delta(\cN,M)-\mathrm{SDP}^\Delta_{n}(\mathcal{N},M)$ holds by construction and thus we consider the upper bound. First, note that again applying \eqref{eq:watrous} we can write
\begin{align}
\mathrm{SDP}^\Delta_n(\mathcal{N},M)=\min&\;\frac{1}{2}\left\|\mathcal{W}(\mathcal{N})_{A\to\bar B}-\mathcal{I}_{A\bar B}\right\|_{\Diamond} \\
\mathrm{s.t.} &\; W_{A\bar A(B\bar B)_1^n}\succeq0,\;\mathrm{Tr}\left[W_{A\bar A(B\bar B)_1^n}\right]=1\\
&\;W_{A\bar A(B\bar B)_1^n}=\mathcal{U}_{(B\bar B)_1^n}^\pi\left(W_{A\bar A(B\bar B)_1^n}\right)\;\forall\pi\in\mathfrak{S}_n\\
&\;W_{A(B\bar B)_1^n}=\frac{1_A}{d_A}\otimes W_{(B\bar B)_1^n},\;W_{A\bar A(B\bar B)_1^{n-1}B_n}=W_{A\bar A(B\bar B)_1^{n-1}}\otimes\frac{1_{B_n}}{d_B}
\end{align}
with the quantum channel $\mathcal{W}(\mathcal{N})_{A\to\bar B}$ defined via its Choi state
\begin{align}
J^{\mathcal{W}(\mathcal{N})}_{ A \bar B}:=d_{\bar A}d_{B}\cdot\mathrm{Tr}_{\bar A B}\left[\left(1_{A} \otimes J^{\mathcal{N}}_{\bar A B}\otimes 1_{\bar B}\right)W_{A\bar A B \bar B}\right].
\end{align}
Second, using the de Finetti Theorem \ref{thm:main-deFinetti} we get that for every feasible Choi state $W_{A\bar A{(B\bar B)}_1^n}$ in $\mathrm{SDP}^\Delta_n(\mathcal{N},M)$, there exists a feasible Choi state $E_{A\bar A}\otimes D_{B\bar B}$ in $\Delta(\cN,M)$ from Lemma \ref{lem:Choi-worst}, such that
\begin{align}
\left\|E_{A\bar A}\otimes D_{B\bar B}-W_{A\bar A B \bar B}\right\|_1\leq\frac{\mathrm{poly}(d)}{\sqrt{n}}.
\end{align}
Third, employing the triangle inequality for the diamond norm we have
\begin{align}
&\left\|\mathcal{D}_{B\to\bar B}\circ\mathcal{N}_{\bar A\to B}\circ\mathcal{E}_{A\to\bar A}-\mathcal{I}_{A\to\bar B}\right\|_{\Diamond}-\left\|\mathcal{W}(\mathcal{N})_{A\to\bar B}-\mathcal{I}_{A\to\bar B}\right\|_{\Diamond}\\
&\leq\left\|\mathcal{D}_{B\to\bar B}\circ\mathcal{N}_{\bar A\to B}\circ\mathcal{E}_{A\to\bar A}-\mathcal{W}(\mathcal{N})_{A\to\bar B}\right\|_{\Diamond}.\label{eq:diamond-triangle}
\end{align}
Forth, relating the trace norm distance of Choi states to the diamond norm distance of quantum channels~\cite[Lemma 7]{Wallman14}, we have
\begin{align}
 \left\|\mathcal{D}_{B\to\bar B}\circ\mathcal{N}_{\bar A\to B}\circ\mathcal{E}_{A\to\bar A}-\mathcal{W}(\mathcal{N})_{A\to\bar B}\right\|_{\Diamond} \leq d_A\cdot\left\|J^{\mathcal{D}\circ\mathcal{N}\circ\mathcal{E}}_{ A \bar B}-J^{\mathcal{W}(\mathcal{N})}_{ A \bar B}\right\|_{1}
\end{align} 
and thanks to the monotonicity under partial trace and H\"older's inequality this bounds \eqref{eq:diamond-triangle} as
\begin{align}
 &\left\|\mathcal{D}_{B\to\bar B}\circ\mathcal{N}_{\bar A\to B}\circ\mathcal{E}_{A\to\bar A}-\mathcal{W}(\mathcal{N})_{A\to\bar B}\right\|_{\Diamond}\\
 &\leq d_Ad_{\bar A}d_{B}\cdot\left\|  \mathrm{Tr}_{\bar A B}\left[\left(1_{A} \otimes J^{\mathcal{N}}_{\bar A B}\otimes 1_{\bar B}\right)\left(E_{A\bar A}\otimes D_{B\bar B}-W_{A\bar A B \bar B}\right)\right] \right\|_{1}\\
 &\leq d_Ad_{\bar A}d_{B}\cdot\left\|\left(1_{A} \otimes J^{\mathcal{N}}_{\bar A B}\otimes 1_{\bar B}\right)(E_{A\bar A}\otimes D_{B\bar B}-W_{A\bar A B \bar B})\right\|_{1} \\
 &\leq d_Ad_{\bar A}d_{B}\cdot\left\|1_{A} \otimes J^{\mathcal{N}}_{\bar A B}\otimes 1_{\bar B}\right\|_{\infty} \left\|E_{A\bar A}\otimes D_{B\bar B}- W_{A\bar A B \bar B}\right\|_{1} \\
 &\leq\frac{\mathrm{poly}(d)}{\sqrt{n}}\quad\text{with $d=\max\{d_A,d_{\bar A},d_B,d_{\bar B}\}$.}
\end{align}
Finally, optimising in \eqref{eq:diamond-triangle} over all feasible Choi states $W_{A\bar A{(B\bar B)}_1^n}$ and then optimising over all feasible Choi states $E_{A\bar A}\otimes D_{B\bar B}$, we get the claimed upper bound
\begin{align}
\Delta(\cN,M)-\mathrm{SDP}^\Delta_{n}(\mathcal{N},M)\leq\frac{\mathrm{poly}(d)}{\sqrt{n}}.
\end{align}
\end{proof}

Numerically, we have found that for the qubit depolarizing channel the first level of our hierarchy already gives the exact optimal value
\begin{align}
\Delta(Dep_2,2)=\mathrm{SDP}^\Delta_{1,\mathrm{PPT}}(Dep_2,2),
\end{align}
which coincides with $1-F(Dep_2,2)$. That is, for the qubit depolarizing channel the average and worst case error criteria become the same.


\section{Distortion with side information}\label{app:lemmas}

The following lemma shows that if the $A$ system is not measured, then the loss in distinguishability after applying a measurement on the $B$ system can be bounded independently of $d_A$.

\begin{lemma}\label{lem:ic-meas-side-info}
Consider a state two-design on $B$, i.e., a set of rank-one projectors $\{P_z\}_{z \in \{1, \dots, t\}}$ such that $\frac{1}{t} \sum_{z=1}^t P_z \otimes P_z = \frac{2 P^{\mathrm{sym}}}{d_B (d_B+1)}$, where $P^{\mathrm{sym}}$ denotes the projector onto the symmetric subspace of $B \otimes B$. Let $\mathcal{M}_B$ be the measurement defined as
\begin{align}
\mathcal{M}_B(X) = \sum_{z} \frac{d_B}{t} \cdot\mathrm{Tr}\big[P_z X\big] |z\rangle\langle z|,
\end{align}
and $\xi_{AB}$ be a Hermitian matrix on $A\otimes B$. Then, we have that
\begin{align}
\| (\mathcal{I}_{A} \otimes \mathcal{M}_B)(\xi_{AB}) \|_1\geq \frac{1}{d_B^2 (d_B+1)} \| \xi_{AB} \|_{1}.
\end{align}
\end{lemma}

We note that the existence of such two-designs is known for any dimension, see e.g., \cite[Corollary 5.3]{scott08} for unitary two-designs and applying these unitaries to any fixed state leads to a state two-design.

\begin{proof}
For any full rank quantum state $\sigma_{A}$, we have by a H\"older type inequality for $\sigma$-weighted Schatten norms that (see, e.g., \cite{OLKIEWICZ1999246} or \cite{beigi13})
\begin{align}
\| (\mathcal{I}_{A} \otimes \mathcal{M}_B)(\xi_{AB}) \|_1 \geq  \frac{\left\| \sigma_{A}^{-1/4} (\mathcal{I}_{A} \otimes \mathcal{M})(\xi_{AB}) \sigma_{A}^{-1/4} \right\|^2_{2}}{\left\| \sigma_{A}^{-1/2} (\mathcal{I}_{A} \otimes \mathcal{M}_B)(\xi_{AB}) \sigma_{A}^{-1/2} \right\|_{\infty}}.
\end{align}
For example, the above inequality can be obtained using \cite[Corollary 3]{beigi13} with the operator $\sigma_{A}^{-1/2} (\mathcal{I}_{A} \otimes \mathcal{M}_B)(\xi_{AB}) \sigma_{A}^{-1/2}$, weight $\sigma$, and $p_\theta = 2$, $\theta = 1/2$, $p_0 = 1$, $p_1 = \infty$. We note that this particular H\"older type inequality for $\sigma$-weighted norms is elementary and follows easily from the usual H\"older inequality, but one way of potentially improving the dimension dependence in Lemma~\ref{lem:ic-meas-side-info} might be to use another H\"older inequality, in particular the $(1,4)$ inequality.

Henceforth, we abbreviate $d\equiv d_B$. To further bound the numerator, letting $\tilde{\xi}_{AB}:=\sigma_{A}^{-1/4} \xi_{AB} \sigma_{A}^{-1/4}$ we get
\begin{align}
\| (\mathcal{I}_{A} \otimes \mathcal{M}_B)(\tilde{\xi}_{AB}) \|^2_{2}
&= \left\| \sum_{z} \frac{d}{t} |z\rangle \langle z| \otimes  \mathrm{Tr}_{B}\left[(1_{A} \otimes P_z) \tilde{\xi}_{AB}\right] \right\|_2^2 \\
&= \sum_{z}  \frac{d^2}{t^2} \mathrm{Tr}\left[\mathrm{Tr}_{B}\left[(1_{A} \otimes P_z) \tilde{\xi}_{AB}\right] \otimes \mathrm{Tr}_{\bar{B}}\left[(1_{\bar{A}} \otimes P_z) \tilde{\xi}_{\bar{A}\bar{B}}\right]^{\dagger} F_{A\bar{A}} \right] \\
&= \sum_{z} \frac{d^2}{t^2} \mathrm{Tr}\left[\left(\left((1_{A} \otimes P_z) \tilde{\xi}_{AB}\right) \otimes \left((1_{\bar{A}} \otimes P_z) \tilde{\xi}_{\bar{A}\bar{B}}\right)^{\dagger}\right)\left(F_{A\bar{A}} \otimes1_{B\bar{B}}\right)\right] \\
&= \frac{d^2}{t^2} \mathrm{Tr}\left[\left(\tilde{\xi}_{AB} \otimes \tilde{\xi}_{\bar{A}\bar{B}}^{\dagger}\right)\left(\sum_{z} (1_{A\bar{A}} \otimes P_z \otimes P_{z})\right)\left(F_{A\bar{A}} \otimes1_{B \bar{B}}\right)\right] \\
&= \frac{1}{t} \frac{d^2}{d(d+1)} \mathrm{Tr}\left[\left(\tilde{\xi}_{AB} \otimes \tilde{\xi}_{\bar{A}\bar{B}}^{\dagger}\right) (1_{A\bar{A}} \otimes (1_{B\bar{B}} + F_{B\bar{B}}))\left(F_{A\bar{A}} \otimes1_{B \bar{B}}\right)\right] \\
&= \frac{1}{t} \frac{d^2}{d(d+1)} \Big(\underbrace{\mathrm{Tr}\left[\tilde{\xi}_{A} \tilde{\xi}_{A}^{\dagger}\right]}_{\geq 0} + \underbrace{\mathrm{Tr}\left[\tilde{\xi}_{AB} \tilde{\xi}_{AB}^{\dagger}\right]}_{=\left\|\tilde{\xi}_{AB}\right\|_2^2} \Big) \\
&\geq \frac{1}{t}\frac{d^2}{d^2 (d+1)} \| \xi_{AB} \|^2_{1} ,
\end{align}
where $F$ denotes the swap operator (as defined in Section \ref{sec:notation}) and in the last step used the H\"older inequality (see, e.g., \cite{bhatia97})
\begin{align}
\left\|\xi_{AB}\right\|_1=\left\|\sigma^{1/4}\sigma^{-1/4}\xi_{AB}\sigma^{-1/4}\sigma^{1/4}\right\|_1\leq\left\|\sigma^{1/4}\otimes1_B\right\|_4\left\|\sigma_A^{-1/4}\xi_{AB}\sigma_A^{-1/4}\right\|_2\left\|\sigma^{1/4}\otimes1_B\right\|_4\leq\sqrt{d}\left\|\tilde{\xi}_{AB}\right\|_2.
\end{align}
For further bounding the denominator we write
\begin{align}
\left\| \sigma_{A}^{-1/2} (\mathcal{I}_{A} \otimes \mathcal{M}_B)(\xi_{AB}) \sigma_{A}^{-1/2} \right\|_\infty
&= \max_{z} \frac{d}{t} \left\| \mathrm{Tr}_{B}\left[(1_{A} \otimes P_z) \sigma_{A}^{-1/2} \xi_{AB} \sigma_{A}^{-1/2}\right]\right\|_\infty \\
&\leq \frac{d}{t} \max_{\ket{\phi}_{A}, \ket{\psi}_{B}} \bra{\phi}_{A} \otimes \bra{\psi}_{B} \sigma_{A}^{-1/2} \xi_{AB} \sigma_{A}^{-1/2} \ket{\phi}_{A} \otimes \ket{\psi}_{B} \\
&\leq \frac{d}{t} \left\| \sigma_{A}^{-1/2} \xi_{AB} \sigma_{A}^{-1/2} \right\|_\infty,
\end{align}
where we used the fact that $P_z$ is a rank 1 projector. Now, observe that for any $\xi_{AB}$, there exists a $\sigma_A$ of unit trace such that
\begin{align}
\frac{\sqrt{\xi_{AB} \xi_{AB}^{\dagger}}}{\| \xi_{AB} \|_{1}} \preceq d\cdot\sigma_{A} \otimes1_{B}.
\end{align}
This just follows from, e.g., \cite[Lemma B.6]{bertachristandl11}, where it is shown that we can in fact choose\footnote{By a continuity argument $\sigma_A$ can be assumed to have full rank.}
\begin{align}
\sigma_{A} = \| \xi_{AB} \|_{1}^{-1}\cdot\mathrm{Tr}_{B}\left[\sqrt{\xi_{AB} \xi_{AB}^{\dagger}}\right].
\end{align}
As a result, we have
\begin{align}
\sqrt{\xi_{AB} \xi_{AB}^{\dagger}} \preceq d\| \xi_{AB} \|_{1}\cdot\sigma_{A} \otimes1_{B}.
\end{align}
As $\xi_{AB}$ is Hermitian, we can decompose it into the positive and negative part $\xi_{AB} = P - Q$ with $P$ and $Q$ positive semidefinite and $PQ = 0$, then $\sqrt{\xi_{AB} \xi_{AB}^{\dagger}} = P + Q$ and so $-\sqrt{\xi_{AB} \xi_{AB}^{\dagger}} \preceq \xi_{AB} \preceq \sqrt{\xi_{AB} \xi_{AB}^{\dagger}}$. Thus, we get
\begin{align}
- d_{B} \| \xi_{AB} \|_1\cdot\sigma_{A} \otimes1_{B} \preceq \xi_{AB} \preceq d_{B} \| \xi_{AB} \|_1 \cdot\sigma_{A} \otimes1_{B},
\end{align}
and we find $\| \sigma_{A}^{-1/2} \xi_{AB} \sigma_{A}^{-1/2} \|_{\infty} \leq d\| \xi_{AB} \|_1$. This concludes the proof.
\end{proof}


\section{Missing proofs}\label{app:missing-proofs}

In the following we give the proofs omitted in the main discussion.

\fdim*

\begin{proof}
The lower bound is trivial and the upper bounds follow directly from the more general statements about the optimal fidelity under additional classical communication assistance as given in Lemma~\ref{lem:F-dimension-locc}.
\end{proof}

\boundsdpn*

\begin{proof}
The lower bound is trivial. By the monotonicity in $n$ (Theorem~\ref{thm:convergence-plain}), it is enough to restrict to $n=1$ for the upper bounds.\footnote{Alternatively, the upper bound of one can directly be deduced operationally from~\cite[Theorem 3]{matthews14}, where $\mathrm{SDP}_1(\mathcal{N},M)$ was identified as the non-signalling assisted channel fidelity.} As in the proof of Lemma~\ref{lem:F-dimension-locc} we mostly use that for any sub-normalized bipartite quantum state $\rho_{XY}$ we have $d_X\cdot1_X\otimes\rho_Y\succeq\rho_{XY}$. For the first upper bound we find $\frac{d_{\bar B}}{d_B}\cdot W_{A\bar A}\otimes1_{B_1\bar B_1}\succeq W_{A\bar AB_1\bar B_1}$, which gives for the objective function
\begin{align}
\mathrm{SDP}_1(\mathcal{N},M)&\leq d_{\bar A}d_B\cdot\mathrm{Tr}\left[\left(J^\mathcal{N}_{\bar AB_1}\otimes\Phi_{A\bar B_1}\right)\left(\frac{d_{\bar B}}{d_B}\cdot W_{A\bar A}\otimes1_{B_1\bar B_1}\right)\right]\\
&=d_{\bar A}d_{\bar B}\cdot\mathrm{Tr}\left[\left(\frac{1_A}{d_A}\otimes\frac{1_{\bar A}}{d_{\bar A}}\right)W_{A\bar A}\right]=\mathrm{Tr}\left[W_{A\bar A}\right]=1.
\end{align}
For the second upper bound we find similarly as for the first upper bound $\frac{d_{\bar A}}{d_A}\cdot1_{A\bar A}\otimes W_{B_1\bar B_1}\succeq W_{A\bar AB_1\bar B_1}$, which then leads to the claim by the same argument as for the second upper bound in Lemma~\ref{lem:F-dimension}.
\end{proof}


\bibliographystyle{abbrv}
\bibliography{library}

\end{document}